\documentclass[onecolumn,draftnocls]{IEEEtran}
\IEEEoverridecommandlockouts

\def\hatx{\hat{x}}

\def\hatv{\hat{v}}

\def\hata{\hat{a}}

\global\long\def\EE{\mathbb{E}}
\global\long\def\PP{\mathbb{P}}

\global\long\def\11{\mathbbm{1}}

\def\ulines{\underline{s}}
\def\ulineS{\underline{S}}

\def\ulinem{\underline{m}}
\def\ulineM{\underline{M}}

\def\ulineX{\underline{X}}
\def\ulinex{\underline{x}}

\def\ulineCalS{\underline{\mathcal{S}}}
\def\ulineCalX{\underline{\mathcal{X}}}

\def\ulineCalV{\underline{\mathcal{V}}}

\def\ulineCalI{\underline{\mathcal{I}}}

\def\ulineY{\underline{Y}}

\def\ulineX{\underline{X}}
\def\ulineY{\underline{Y}}
\def\ulineV{\underline{V}}
\def\ulinev{\underline{v}}
\def\ulinea{\underline{a}}

\def\tildea{\tilde{a}}
\def\tildem{\tilde{m}}

\def\3To1BC{$3-$to$-1$}

\def\define{:{=}~}

\def\naturals{\mathbb{N}}

\def\hatm{\hat{m}}

\def\fieldq{\mathcal{F}_{q}}

\newif\ifProofForORDBC

\def\parsec{\par\noindent}
\def\med{\medskip\parsec}

\newif\ifJournal

\usepackage{amssymb}
\usepackage{amsmath}
\usepackage{mathrsfs}
\usepackage{ulem}
\usepackage{epsf,epsfig}
\usepackage{cite}
\usepackage{color}
\usepackage{dsfont}
\usepackage{bbm}
\usepackage{amsthm}
\usepackage{physics}
\newtheorem{theorem}{Theorem}

\def\ulines{\underline{s}}
\def\ulineS{\underline{S}}

\def\ulineX{\underline{X}}
\def\ulinex{\underline{x}}

\def\ulineCalS{\underline{\mathcal{S}}}
\def\ulineCalX{\underline{\mathcal{X}}}

\def\ulineCalV{\underline{\mathcal{V}}}

\def\ulineY{\underline{Y}}

\def\ulineX{\underline{X}}
\def\ulineY{\underline{Y}}
\def\ulineV{\underline{V}}
\def\ulinev{\underline{v}}
\def\ulinea{\underline{a}}

\def\CalF{\mathcal{F}}
\def\CalJ{\mathcal{J}}

\def\CalI{\mathcal{I}}
\def\CalP{\mathcal{P}}
\def\CalS{\mathcal{S}}
\def\CalM{\mathcal{M}}
\def\CalV{\mathcal{V}}
\def\CalX{\mathcal{X}}

\def\EE{\mathbb{E}}

\def\PP{\mathbb{P}}
\def\WW{\mathbb{W}}
\def\11{\mathbbm{1}}

\def\TDelta{\mathcal{T}_{\delta}^{(n)}}

\usepackage{stackengine}
\def\deq{\mathrel{\ensurestackMath{\stackon[1pt]{=}{\scriptstyle\Delta}}}}
\def\define{\mathrel{\ensurestackMath{\stackon[1pt]{=}{\scriptstyle\Delta}}}}

\newcommand{\comment}[1]{}

\begin{document}

\sloppy
\newtheorem{remark}{\it Remark}
\newtheorem{thm}{Theorem}
\newtheorem{corollary}{Corollary}
\newtheorem{definition}{Definition}
\newtheorem{lemma}{Lemma}
\newtheorem{example}{Example}
\newtheorem{prop}{Proposition}
\title{\huge Computing Sum of Sources over a Classical-Quantum MAC}
\author{\IEEEauthorblockN{Touheed Anwar Atif\IEEEauthorrefmark{1},
		Arun Padakandla\IEEEauthorrefmark{2} and  S. Sandeep Pradhan\IEEEauthorrefmark{1} \\}
	\IEEEauthorblockA{Department of Electrical Engineering and Computer Science,\\
		\IEEEauthorrefmark{1}University of Michigan, Ann Arbor, USA.\\
		\IEEEauthorrefmark{2}University of Tennessee, Knoxville, USA\\
		Email: \tt touheed@umich.edu, arunpr@utk.edu, pradhanv@umich.edu}}

\maketitle

\begin{abstract}
We consider the problem of communicating a general bivariate function of two classical sources observed at the encoders of a classical-quantum multiple access channel. Building on the techniques developed for the case of a classical channel, we propose and analyze a coding scheme based on coset codes. The proposed technique enables the decoder recover the desired function without recovering the sources themselves. We derive a new set of sufficient conditions that are weaker than the current known for identified examples. This work is based on a new ensemble of coset codes that are proven to achieve the capacity of a classical-quantum point-to-point channel.
\end{abstract}

\section{Introduction}
\label{sec:intro}

Early research in quantum state discrimination led to the investigation of the information carrying capacity of quantum states. Suppose Alice - a sender - can prepare any one of the states in the collection $\{\rho_{x} \in \mathcal{D}(\mathcal{H}_{Y}): x \in \mathcal{X}\}$ and Bob - the receiver - has to rely on a measurement to infer the label $x$ of the state, then what is the largest sub-collection $\mathcal{C} \subseteq \mathcal{X}$ of states that Bob can distinguish perfectly? Studying this question in a Shannon-theoretic sense, Schumacher, Westmoreland \cite{199707PhyRev_SchWes} and Holevo \cite{199801TIT_Hol} characterized the exponential growth of this sub-collection, thereby characterizing the capacity of a classical-quantum (CQ) point-to-point (PTP) channel. In the following years, generalizations of this question with multiple senders and/or receivers have been studied with an aim of characterizing the corresponding information carrying capacity of quantum states in network scenarios \cite{winter2001capacity}. 

In this work, we consider the problem of computing functions of information sources over a CQ multiple access channel (MAC). Let $(\rho_{x_{1}x_{2}} \in \mathcal{D}(\mathcal{H}_{Y}): (x_{1},x_{2}) \in \mathcal{X}_{1} \times \mathcal{X}_{2} )$ model a CQ-MAC. Sender $j$ - the party having access to the choice of label $x_{j} \in \mathcal{X}_{j}$ - observes a classical information stream $S_{jt} \in \mathcal{S}_{j}: t \geq 1$. The pairs $(S_{1t},S_{2t}) : t \geq 1$ are independent and identically distributed (IID) with a single-letter joint distribution $\mathbb{W}_{S_{1}S_{2}}$. The receiver, who is provided with the prepared quantum state, intends to reconstruct a specific function $f(S_{1},S_{2})$ of the information observed by the senders. The question of interest is under what conditions, specified in terms of the CQ-MAC, $\mathbb{W}_{S_{1}S_{2}}$ and $f$, can the receiver reconstruct the desired function losslessly? 

The conventional approach to characterizing sufficient conditions for this problem relies on enabling the receiver reconstruct the pair of classical source sequences. Since the receiver is only interested in recovering the bivariate function $f$, and not the pair, this approach can be strictly sub-optimal. Can we exploit this and design a more efficient communication strategy, thereby weakening the set of sufficient conditions? In this work, we present one such communication strategy for a general CQ-MAC that is more efficient than the conventional approach. This strategy is based on 
asymptotically good random nested coset codes. 
We analyze its performance and derive new sufficient conditions for a general problem instance and identify examples for which the derived conditions are strictly weaker.

\begin{figure}
    \centering
    \includegraphics[width=3.7in]{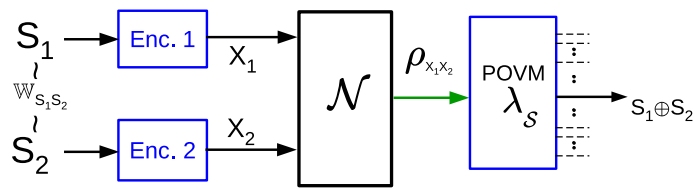}
    \caption{CQ-MAC used for computing sum of classical sources.}
    \label{Fig:ComputeSumOfSources}
\end{figure}
Our findings here are built on the ideas developed in the classical setting. Focusing on a source coding formulation, i.e. a noiseless MAC, K\"orner and Marton \cite{197903TIT_KorMar} devised an ingenious coding technique that enabled the receiver recover the sum of the sources without recovering either source. In \cite{200710TIT_NazGas}, the linearity of the K\"orner-Marton (KM) source coding map was further exploited to enable the receiver recover the sum of the sources using \textit{only} the \textit{sum} of the KM indices, not even requiring the pair. Leveraging this observation and focusing on the subclass of additive MACs, specific MAC channel coding techniques are devised in \cite{200710TIT_NazGas} that enabled the receiver recover the sum of two channel coding message indices. 

The techniques of \cite{197903TIT_KorMar}, \cite{200710TIT_NazGas} are instances of a broader framework of coding strategies. Decoding functions of sources or channel inputs efficiently require codes endowed with algebraic closure properties. To emphasize, the conventional approach of deriving inner bounds/achievable rate region by analyzing expected performance of IID random codes is incapable of yielding performance limits - capacity or rate-distortion regions as the case may be- in network communication scenarios. To improve upon this, it is necessary to analyze the expected performance of random codes endowed with algebraic closure properties. In a series of works \cite{2020Bk_PraPadShi}, an information theoretic study of the latter codes has been carried out yielding new inner bounds for multiple network communication scenarios. 

In this work, we embark on developing these ideas in the CQ setup. After having provided the problem statement in Sec.~\ref{sec:prelims}, we focus on a simplified CQ MAC and illustrate the core idea of our coding scheme. The latter relies on developing a \textit{nested coset code} (NCC) based communication scheme for a CQ PTP channel and analyzing its performance (Sec.~\ref{Sec:NCCAchieveCQ-PTPCapacity}). Leveraging this building block, we design and analyze the performance of an NCC-based coding scheme for computing sum over a general CQ-MAC (Sec.~\ref{sec:arbitraryCQ-MAC}). Going further we generalize this idea for computing arbitrary functions over a general CQ-MAC.

In this work, we consider a classical-quantum generalization of the 

that Alice prepared. What is the size of the largest sub-collection that can be perfectly distinguished by Bob? 

capacity of a classical-quantum (CQ) (point-to-point (PTP)) channel. 

In this case, how many states can Alice prepare that can be perfectly distinguished by Bob? Studying this question in a Shannon-theoretic sense, Schumacher, Westmoreland and Holevo characterized the capacity of a classical-quantum (CQ) (point-to-point (PTP)) channel. These and other related findings have led to a systematic study of the information carrying capacity of CQ channels with multiple senders and/or receivers. While being aided by the ideas developed in network information theory, this study is met with a unique set of challenges owing to non-commutative nature of the underlying operator algebra. Indeed, quantum analogues of even well-established ideas such as joint decoding remains to be discovered.

Early research in quantum state discrimination led to the investigation of the information carrying capacity of quantum states. Suppose Alice - a sender - can prepare any one of the states in the collection $\{\rho_{x} : x \in \mathcal{X}\}$ and Bob - the receiver - has to rely on a measurement to infer the label $x$ of the state that Alice prepared. In this case, how many states can Alice prepare that can be perfectly distinguished by Bob? Studying this question in a Shannon-theoretic sense, Schumacher, Westmoreland and Holevo characterized the capacity of a classical-quantum (CQ) (point-to-point (PTP)) channel. These and other related findings have led to a systematic study of the information carrying capacity of CQ channels with multiple senders and/or receivers. While being aided by the ideas developed in network information theory, this study is met with a unique set of challenges owing to non-commutative nature of the underlying operator algebra. Indeed, quantum analogues of even well-established ideas such as joint decoding remains to be discovered.

Our pursuit of characterizing capacity regions of multi-terminal networks

Our findings build upon a new ensemble of codes for communication over classical-quantum channels that yield strictly better performance for network scenarios.

}\end{comment}
\section{Preliminaries and Problem Statement}
\label{sec:prelims}
We supplement the notation in \cite{2013Bk_Wil} with the following. For positive integer $n$, $[n] \define \left\{1,\cdots,n \right\}$. For a Hilbert space $\mathcal{H}$, $\mathcal{L}(\mathcal{H}),\mathcal{P}(\mathcal{H})$ and $\mathcal{D}(\mathcal{H})$ denote the collection of linear, positive and density operators acting on $\mathcal{H}$, respectively. 
The von Neumann entropy of a density operator $\rho \in \mathcal{D}(\mathcal{H})$ is denoted by $S(\rho)$. 
Given any ensemble $\{p_i, \rho_i\}_{i\in [1,m]}$, the Holevo information \cite{holevo} is denoted as $\chi \big( \{p_i; \rho_i\}\big)$.  A POVM acting on $\mathcal{H}$ is a collection $\lambda_\mathcal{X}\deq \{\lambda_x\}_{x \in \mathcal{X}}$ of positive operators that form a resolution of the identity: $\sum_{x \in \mathcal{X}} \lambda_x=I$, where $\mathcal{X}$ is a finite set.
We employ an \underline{underline} notation to aggregate objects of similar type. For example, $\ulines$ denotes $ (s_{1},s_{2})$, $\ulinex^{n}$ denotes $(x_{1}^{n},x_{2}^{n})$, $\ulineCalS$ denotes the Cartesian product $\mathcal{S}_{1}\times \mathcal{S}_{2}$.

Consider a (generic) \textit{CQ-MAC} $(\rho_{x_{1}x_{2}} \in \mathcal{D}(\mathcal{H}_{Y}): (x_{1},x_{2}) \in \mathcal{X}_{1}\times \mathcal{X}_{2})$ specified through (i) finite sets $\mathcal{X}_{j} : j \in [2]$, (ii) Hilbert space $\mathcal{H}_{Y}$, and (iii) a collection $( \rho_{x_{1},x_{2}} \in \mathcal{D}(\mathcal{H}_{Y} ) :  (x_{1},x_{2}) \in \mathcal{X}_{1}\times \mathcal{X}_{2})$ of density operators. This CQ-MAC is employed to enable the receiver reconstruct a bivariate function of the classical information streams observed by the senders. Let $\mathcal{S}_{1},\mathcal{S}_{2}$ be finite sets and $(S_{1},S_{2}) \in \mathcal{S}_{1}\times \mathcal{S}_{2}$ distributed with PMF $\mathbb{W}_{S_{1}S_{2}}$ models the pair of information sources observed at the encoders. Specifically, sender $j$ observes the sequence $S_{jt} \in \mathcal{S}_{j}: t \geq 1$ and the sequence $(S_{1t},S_{2t}): t \geq 1$ are IID with single-letter PMF $\mathbb{W}_{S_{1}S_{2}}$. The receiver aims to recover the sequence $f(S_{1t},S_{2t}) : t \geq 1$ losslessly, where $f:\mathcal{S}_{1}\times \mathcal{S}_{2} \rightarrow \mathcal{R}$ is a specified function.

A CQ-MAC code $c_{f}=(n,e_{1},e_{2},\lambda_{\mathcal{R}^{n}})$ of block-length $n$ for recovering $f$ consists of two encoders maps $e_{j} : \mathcal{S}^{n} \rightarrow \mathcal{X}_{j}^{n} : j \in [2]$, and a POVM $\lambda_{\mathcal{R}^{n}}= \{ \lambda_{r^{n}} \in \mathcal{P}(\mathcal{H}_{Y}^{\otimes n}) : r^{n} \in \mathcal{R}^{n}\}$. The average error probability of the CQ-MAC code $c_{f}$ is
\begin{eqnarray}
\label{Eqn:ProbStatementGeneralfrrorOfCode}
\overline{\xi}(c_{f}) &=& 1- \sum_{\ulines^{n}:f(\ulines^{n})=r^{n}}\mathbb{W}_{S_{1}S_{2}}^{n}(s_{1}^{n},s_{2}^{n})\tr(\lambda_{r^{n}}\rho^{\otimes n}_{c,\ulines^{n}})
 \nonumber
\end{eqnarray}
where $\rho^{\otimes n}_{c,\ulines^{n}} \define \otimes_{i=1}^{n}\rho_{x_{1i}(s_{1}^{n})x_{2i}(s_{2}^{n})}$, where $e_{j}(s_{j}^{n}) = x_{j1}(s_{j}^{n}),x_{j2}(s_{j}^{n}),\cdots, x_{jn}(s_{j}^{n})$ for $j\in [2]$.

A function $f$ of the sources $\mathbb{W}_{S_{1}S_{2}}$ is said to be reconstructible over a CQ-MAC if for $\epsilon > 0$, $\exists$ a sequence $c_{f}^{(n)} = (n,e_{1}^{(n)},e_{2}^{(n)},\lambda_{\mathcal{R}^{n}}) : n \geq 1$ such that $\lim_{n \rightarrow \infty} \overline{\xi}(c_{f}^{(n)}) = 0$.

In this article, we are concerned with the problem of characterizing sufficient conditions under which a function $f$ of the sources $\mathbb{W}_{S_{1}S_{2}}$ is reconstructible over a generic MAC $(\rho_{x_{1}x_{2}} \in \mathcal{D}(\mathcal{H}_{Y}): (x_{1},x_{2}) \in \mathcal{X}_{1}\times \mathcal{X}_{2})$. One of our findings - Proposition \ref{Prop:GeneralFnOverArbCQMAC} - provides a characterization of sufficient conditions in terms of a computable function of the associated objects- density operators that characterize the CQ-MAC, function $f$ and the source distribution $\mathbb{W}_{S_{1}S_{2}}$.

As we shall see, the specific problem of computing sum of sources will play an important role in our work. In this case, $\mathcal{S}=\mathcal{S}_{1}=\mathcal{S}_{2}=\mathcal{F}_{q}$ is a finite field with $q$ elements and the receiver aims to reconstruct $f(S_{1},S_{2})=S_{1}\oplus_{q}S_{2}$ where $\oplus_{q}$ denotes addition in $\mathcal{F}_{q}$. A CQ-MAC code $c_{\oplus}=(n,e_{1},e_{2},\lambda_{\mathcal{S}^{n}})$ of block-length $n$ for recovering the sum consists of two encoders maps $e_{j} : \mathcal{S}^{n} \rightarrow \mathcal{X}_{j}^{n} : j \in [2]$, and a POVM $\lambda_{\mathcal{S}^{n}}=\{ \lambda_{s^{n}} \in \mathcal{P}(\mathcal{H}_{Y}^{\otimes n}) : s^{n} \in \mathcal{S}^{n}\}$. 

Restricting $f$ to a sum, we say the sum of sources $\mathbb{W}_{S_{1}S_{2}}$ over field $\mathcal{F}_{q}$ is reconstructible over a CQ-MAC if $\mathcal{S}_{1}=\mathcal{S}_{2}=\mathcal{F}_{q}$ and the function $f(S_{1},S_{2})=S_{1}\oplus_{q}S_{2}$ is reconstructible over the CQ-MAC. The problem of characterizing sufficient conditions under which a sum of sources is reconstructible over a CQ-MAC plays an important role in this work. One of our findings - Theorem \ref{thm:SumCQ_MAC} - provides a computable characterization of a set of sufficient conditions under which a sum of sources is reconstructible over a CQ-MAC. As the reader will note, this encapsulates the central element of our characterization in Proposition \ref{Prop:GeneralFnOverArbCQMAC}.


We also formalize the notions of a CQ-PTP and CQ-MAC codes for communicating uniform messages. A CQ-MAC code $c_{\ulinem}=(n,\mathcal{I}_{1},\mathcal{I}_{2},e_{1},e_{2},\lambda_{\ulineCalI})$ for a CQ-MAC $(\rho_{\ulinex} \in \mathcal{D}(\mathcal{H}_{Y}): \ulinex \in \ulineCalX)$ consists of (i) index sets $\mathcal{I}_{j} : j \in [2]$, (ii) encoder maps $e_{j}:\mathcal{I}_{j}\rightarrow \mathcal{X}_{j}^{n} : j \in [2]$ and a decoding POVM $\lambda_{\ulineCalI} = \{ \lambda_{\ulinem} \in \mathcal{P}(\mathcal{H}_{Y}^{\otimes n}): \ulinem \in \mathcal{I}_{1}\times \mathcal{I}_{2} \}$. For $\ulinem \in \mathcal{I}_{1}\times \mathcal{I}_{2}$, we let $\rho^{\otimes n}_{c,\ulinem}\define\otimes_{i=1}^{n}\rho_{x_{1i},x_{2i}}$ where $e_{j}(m_{j})=x_{j1}\cdots x_{jn}$ for $j \in [2]$.

A CQ-PTP code $c_{m}=(n,\mathcal{I},e,\lambda_{\mathcal{I}})$ for a CQ-PTP $(\rho_{x} \in \mathcal{D}(\mathcal{H}_{Y}):x \in \mathcal{X})$ consists of (i) an index set $\mathcal{I}$, (ii) and encoder map $e:\mathcal{I}\rightarrow \mathcal{X}^{n}$ and a decoding POVM $\lambda_{\mathcal{I}} = \{ \lambda_{m} \in \mathcal{P}(\mathcal{H}_{Y}^{\otimes n}): m \in \mathcal{I} \}$. For $m \in \mathcal{I}$, we let $\rho_{c,m}^{\otimes n} \define \otimes_{i=1}^{n}\rho_{x_{i}}$ where $e(m) = x_{1}\cdots x_{n}$.

\section{The Central Idea}
\label{Sec:CentralIdea}
Let us consider the specific problem of reconstructing the sum of sources each taking values in $\mathcal{S} = \mathcal{F}_{q}$. We begin by reviewing the KM coding scheme for the case of a noiseless classical MAC. It was shown in \cite{197903TIT_KorMar} the existence of linear code with a parity matrix $H \in \mathcal{S}^{l \times n}$ and decoder map $d:\mathcal{F}_{q}^{l} \rightarrow \mathcal{S}^{n}$ such that
  $\sum_{\ulines^{n} \in \ulineCalS^{n}} \mathbb{W}_{\ulineS}^{n} (\ulines^{n})\mathds{1}_{\{d(Hs_{1}^{n} \oplus_{q} Hs_{2}^{n}) \neq s_{1}^{n}\oplus_{q} s_{2}^{n}\}} \leq \epsilon$,
for any $\epsilon > 0$, and sufficiently large $n$, so long as $ \frac{l\log_{2}q}{n} > H(S_{1}\oplus_{q} S_{2})$. This implies that a receiver equipped with the decoding map $d$ can recover the sum if it possesses the sum $M_{1}^{l}\oplus_{q}M_{2}^{l}$ of the K\"orner-Marton indices $M_{j}^{l} = HS_{j}^{l} : j \in [2]$.

We are therefore led to building an efficient CQ-MAC coding scheme that enables the receiver only reconstruct the sum of the two message indices. Indeed, if the two senders send the KM indices to such a CQ-MAC channel code and the receiver employs the above source decoder $d$ on the decoded sum of the KM indices, it can recover the sum of sources. 
To illustrate the design of the desired CQ-MAC channel code, let us consider a CQ-MAC $(\rho_{x_{1}x_{2}} \in \mathcal{D}(\mathcal{H}_{Y}): (x_{1},x_{2}) \in \mathcal{X}_{1}\times \mathcal{X}_{2})$ wherein $\mathcal{X}_{1}=\mathcal{X}_{2}=\mathcal{F}_{q}$ and the collection $\rho_{\ulinex} : \ulinex \in \ulineCalX$ satisfies $\rho_{x_{1}x_{2}}=\rho_{\hatx_{1}\hatx_{2}}$ whenever $x_{1}\oplus_{q} x_{2} = \hatx_{1}\oplus_{q}\hatx_{2}$. Consider a CQ-PTP $(\mathcal{X}=\mathcal{F}_{q},\sigma_{u}:u \in \mathcal{X} )$ where $\sigma_{u}=\rho_{x_{1}\oplus x_{2}}$ for any $(x_{1},x_{2})$ satisfying $x_{1}\oplus_{q} x_{2}=u$. Suppose we are able to communicate over this CQ-PTP via a \textit{linear CQ-PTP code} $\mathcal{C} \subseteq \mathcal{X}^{n}$. Specifically, suppose there exists a generator matrix $G \in \mathcal{X}^{l \times n}$ and a POVM $\{ \lambda_{m^{l}} : m^{l} \in \mathcal{F}_{q}^{l} \}$ so that $ 1-q^{-l}\sum_{m^{l}}\tr(\lambda_{m^{l}}\sigma^{\otimes n}_{m^{l}G}) \leq \epsilon$.
for any $\epsilon > 0$ and sufficiently large $n$, where $\sigma^{\otimes}_{m^{l}G} = \sigma_{x_{1}}\otimes \cdots \otimes \sigma_{x_{n}}$ where $m^{l}G = x^{n}$. We can then use this linear CQ-PTP code as our desired CQ-MAC channel code. Indeed, observe that, suppose both senders employ this same linear CQ-PTP code, then sender $j$ maps its KM index $m_{j}^{l} = Hs_{j}^{n}$ to the channel codeword $x_{j}^{n}=m_{j}^{l}G$. Observe that the structure of the CQ-MAC implies $\rho^{\otimes n}_{x_{1}^{n},x_{2}^{n}} =\sigma^{\otimes n}_{x_{1}^{n} \oplus x_{2}^{n}} = \sigma^{\otimes}_{(m_{1}^{l}\oplus_{q} m_{2}^{l})G}$. If the receiver employs the POVM $\{ \lambda_{m^{l}} : m^{l} \in \mathcal{F}_{q}^{l} \}$ designed for the CQ-PTP, it ends up decoding the sum of the KM indices $m_{1}^{l}\oplus_{q} m_{2}^{l}$, and consequently,  recover the sum of the sources.

A careful analysis of the above idea reveals that two MAC channel codes employed by the encoders do not `blow up' when added, is crucial to the efficiency of the above scheme. A linear code being algebraically closed enables this. However, the codewords of a random linear code are uniformly distributed and cannot achieve the capacity of an arbitrary classical PTP channel, let alone a CQ-PTP channel. We are therefore forced to enlarge a linear code to identify sufficiently many codewords of the desired empirical distribution. We are thus led to a \textit{nested coset code} (NCC)\cite{201301arXivComputation_PadPra}. A NCC comprises of cosets of a coarse linear code within a fine code. Within each coset, we can identify a codeword of the desired empirical distribution. We choose as many cosets as the number of messages.
Analogous to our illustration above where we chose a linear code that achieves the capacity of the CQ-PTP $(\mathcal{X}=\mathcal{F}_{q},\sigma_{u}:u \in \mathcal{X})$, our first step (Sec.~\ref{Sec:NCCAchieveCQ-PTPCapacity}) is to design a NCC with its POVM that can achieve capacity of an arbitrary CQ PTP. Our second step is to endow both senders with this same NCC and analyze decoding the sum of the messages. This gets us to our next challenge - How do we analyze decoding their message sum, for a general CQ-MAC $\rho_{\ulinex} : \ulinex \in \ulineCalX$ for which $x_{1}\oplus_{q} x_{2} = \hatx_{1}\oplus_{q}\hatx_{2}$ does \textit{not} necessarily imply $\rho_{x_{1}x_{2}}=\rho_{\hatx_{1}\hatx_{2}}$. In Sec.~\ref{sec:theorem2}, we address this challenge, leverage our findings in Sec.~\ref{Sec:NCCAchieveCQ-PTPCapacity} and generalize the idea
for a general CQ-MAC.

\section{Nested Coset Codes Achieve Capacity of CQ-PTP}
\label{Sec:NCCAchieveCQ-PTPCapacity}
We begin by formalizing the structure of an NCC.

\begin{definition}
 \label{Defn:NCC}
 An $(n,k,l,g_{I},g_{O/I},b^{n})$ NCC built over a finite field $\mathcal{V}=\mathcal{F}_{q}$ comprises of (i) generator matrices $g_{I} \in \mathcal{V}^{k \times n}$, $g_{O/I} \in \mathcal{V}^{l \times n}$ (ii) a bias vector $b^{n}$, an encoder map $e :\mathcal{V}^{l} \rightarrow \mathcal{V}^{k}$. We let $v^{n}(a,m) = ag_{I}\oplus_{q}mg_{O/I}\oplus_{q}b^{n}: (a,m) \in \mathcal{V}^{k} \times \mathcal{V}^{l}$ denote elements in the range space of the generator matrix $[g_{I}^{t}~ g_{O/I}^{t}]^{t}$.
\end{definition}

\begin{definition}
A CQ-PTP code $(n,\mathcal{I}=\mathcal{F}_{q}^{l},e,\lambda_{\mathcal{I}})$ is an NCC CQ-PTP if there exists an $(n,k,g_I,g_{O/I},b^n)$ NCC such that $e(m) \in \{u^{n}(a,m) : a \in \CalF_q^k\}$ for all $m \in \CalF_q^l$.
\end{definition}

\begin{theorem}
\label{Thm:NCCAchievesCQPTP}
 Given a CQ-PTP $(\rho_{v} \in \mathcal{D}(\mathcal{H}_{Y}): v \in \CalF_q)$ and a PMF $p_{V}$ on $\mathcal{F}_{q}$, $\epsilon > 0$ there exists a CQ-PTP code $c=(n,\mathcal{I}=\mathcal{F}_{q}^{l},e,\lambda_{\mathcal{I}})$ such that (i) $q^{-l}\sum_{m \in [\mathcal{I}]}\sum_{\hatm \neq [\mathcal{I}]\setminus\{m\}}\tr(\lambda_{\hatm} \rho^{\otimes n}_{c,m}) \leq \epsilon$, (ii) $c=(n,\mathcal{I}=\mathcal{F}_{q}^{l},e,\lambda_{\mathcal{I}})$ is a NCC CQ-PTP, (iii) $\frac{k\log_{2}q}{n} >  \log_{2}q - H(V)$ and $\frac{(k+l)\log_{2}q}{n} < \log_2{q} - H(V) +\chi(\{p_v,\rho_v\}) $ for all $n$ sufficiently large.
\end{theorem}

\begin{proof}
In order to achieve a rate $R=\chi(\{p_v,\rho_v\})$, the standard approach is to pick $2^{nR}$ codewords uniformly and independently from $T_{\delta}^{n}(p_{V})$. However, the resulting code is not algebraically closed. On the other hand, if we pick a random generator matrix $G \in \CalF_q^{l\times n}$, with $l = \frac{nR}{\log_{2}q}$, whose entries from $\mathcal{F}_{q}$ are IID uniform, then its range space - the resulting collection of $2^{nR}$ codewords - are uniformly distributed and pairwise independent but not $p_{V}-$typical. 

To satisfy the dual requirements of algebraically closure and $p_{V}-$typicality, we observe the following. If a collection of $q^{k}$ codewords are uniformly distributed in $\mathcal{F}_{q}^{n}$ and pairwise independent, as we found the range space of $G$ to be, then the expected number of codewords that are $p_{V}-$typical is $\frac{q^{k}}{q^{n}}|T_{\delta}^{n}(p_{V})| = \exp\{n\log_{2}q \left(\frac{k}{n} -\left[1-\frac{H(V)}{\log_{2}q}\right]\right)\}$. This indicates that if we pick a generator matrix $G_{I} \in \mathcal{F}^{k \times n}$ with entries uniformly distributed and IID, such that $\frac{k}{n} > 1-\frac{H(V)}{\log_{2}q}$, then its range space will contain codewords that are $p_{V}$-typical. The latter codewords can be used for communication.

Each coset of $G_{I} \in \mathcal{F}^{k \times n}$ where $\frac{k}{n} > 1-\frac{H(V)}{\log_{2}q}$ will play an analogous role as a single codeword in a conventional IID random code. Just as we pick $2^{nR}$ of the latter, we consider $2^{nR}$ cosets of $G_{I}$ within a larger linear code with generator matrix $G = \left[ \begin{array}{c}G_{I}\\G_{O/I}\end{array}\right] \in \mathcal{F}_{q}^{(k+l)\times n}$ with $l = \frac{nR}{\log_{2}q}$. The messages index the $2^{nR}$ cosets of $G_{I}$. A predetermined element in each coset that is $p_{V}-$typical is the assigned codeword for the message and chosen for communication.\footnote{The reader is encouraged to relate to the bounds stated in theorem statement and induced bounds on the rate of communication $\frac{l\log_{2}q}{n}$.} A formal proof we provide below has two parts - error probability analysis for a generic fixed code followed by an upper bound on the latter via code randomization.

\med\textit{Upper bound on Error Prob. for a generic fixed code} : Consider a generic NCC $(n,k,l,g_{I},g_{O/I},b^{n})$ with its range space $v^{n}(a,m) = ag_{I}\oplus_{q}mg_{O/I}\oplus_{q}b^{n}: (a,m) \in \mathcal{V}^{k} \times \mathcal{V}^{l}$. We shall use this and define a CQ-PTP code $(n,\mathcal{I}=\mathcal{F}_{q}^{l},e,\lambda_{\mathcal{I}})$ that is an NCC CQ-PTP. Towards that end, let $\theta(m) \define \sum_{a \in \mathcal{V}^{k}}\mathds{1}_{\left\{ v^{n}(a,m) \in T_{\delta}^{n}(p_{V}) \right\}}$ and 
\begin{eqnarray}
\label{Eqn:NoTypicalElements}
 s(m) \define \begin{cases} \{a \in \mathcal{V}^{K}: v^{n}(a,m) \in T_{\delta}^{n}(p_{V})\}&\mbox{if }\theta(m) \geq 1 \\\{0^{k}\}&\mbox{if }\theta(m) = 0,\end{cases}
 \nonumber
\end{eqnarray}
for each $m \in \mathcal{V}^{l}$. For $m \in \mathcal{V}^{l}$, a predetermined element $a_{m} \in s(m)$ is chosen. On receiving message $m \in \mathcal{V}^{l}$, the encoder prepares the quantum state $\rho_{m}^{\otimes n} \define \rho^{\otimes n}_{v^{n}(a_{m},m)} \define \otimes_{i=1}^{n}\rho_{v_{i}(a_{m},m)}$ and is communicated. The encoding map $e$ is therefore determined via the collection $(\!a_{m}\! \in s(m)\!:\! m \in\! \mathcal{V}^{l})$.

Towards specifying the decoding POVM let $\rho_{v} = \sum_{y \in \mathcal{Y}}p_{Y|V}(y|v)\ket{e_{y|v}}\bra{e_{y|v}}$ be a spectral decomposition for $v \in \mathcal{V}$. We let $p_{VY}\define p_{V}p_{Y|V}$. For any $v^{n}\in \CalV^{n}$, let $\pi_{v^{n}}$ be the conditional typical projector as in \cite[Defn. 15.2.4]{2013Bk_Wil} with respect to the ensemble $\{ \rho_{v}:v \in \mathcal{V}\}$ and distribution $p_{V}$. Similarly, let $\pi_{\rho}$ be the (unconditional) typical projector of the state $\rho\define \sum_{v \in \mathcal{V}}p_{V}(v)\rho_{v}$ as defined in \cite[Defn. 15.1.3]{2013Bk_Wil}.
For $(a,m) \in \CalV^{k}\times \CalV^{l}$, we let $\pi_{a,m} \define \pi_{v^{n}(a,m)}\mathds{1}_{\{ v^{n}(a,m) \in T_{\delta}^{n}(p_{V}) \}}$. We let $\lambda_{\mathcal{I}} \define \{ \sum_{a \in \CalV^{k}}\lambda_{a,m} : m \in \mathcal{I}=\CalV^{l},\lambda_{-1} \}$, where 
\begin{align}
 \label{Eqn:CQPTPPOVMDefn}
 \lambda_{a,m} \!\define\! \Big(\! \sum_{\hata \in \CalV^{k}}\! \sum_{\hatm \in \CalV^{l}} \!\!\gamma_{\hata,\hatm}\Big)^{-{1}/{2}}\!\!\! \gamma_{a,m}\Big(\! \sum_{\tildea \in \CalV^{k}} \sum_{\tildem \in \CalV^{l}}\!\! \gamma_{\tildea,\tildem}\Big)^{-{1}/{2}},
\end{align}
$\lambda_{-1}\define I-\sum_{m \in \CalV^{l}}\sum_{a \in \CalV^{k}}\lambda_{a,m}$ and $\gamma_{a,m} \define \pi_{\rho}\pi_{a,m}\pi_{\rho}$. Since $0 \leq \gamma_{a,m} \leq I$, we have $0 \leq \lambda_{a,m} \leq I$. The latter lower bound implies $\lambda_{\mathcal{I}} \subseteq \mathcal{P}(\mathcal{H})$. The same lower bound coupled with the definition of the generalized inverse implies $I \geq \sum_{a \in \CalV^{k}}\sum_{m \in \CalV^{l}}\lambda_{a,m} \geq 0$. We thus have $0 \leq \lambda_{-1} \leq I$. It can now be verified that $\lambda_{\mathcal{I}}$ is a POVM. In essence, the elements of this POVM is identical to the standard POVMs 
except the POVM elements corresponding to a coset have been added together. Indeed, since each coset corresponds to one message, there is no need to disambiguate within the coset.

We have thus associated an NCC $(n,k,l,g_{I},g_{O/I},b^{n})$ and a collection $(a_{m} \in s(m): m \in \mathcal{V}^{l})$ with a CQ-PTP code. 
The error probability of this code is
\begin{eqnarray}
 \label{Eqn:CQ-PTPErrorProbability}
 q^{-l}\!\!\sum_{m \in \mathcal{I}}\!\mbox{tr}((I-\sum_{a \in \CalV^{k}}\!\!\lambda_{a,m})\rho_{m}^{\otimes n}) \leq  q^{-l}\!\!\sum_{m \in \mathcal{I}}\!\mbox{tr}((I-\lambda_{a_{m},m})\rho_{m}^{\otimes n}).
\end{eqnarray}
Denoting event $\mathscr{E}=\{ \theta(m)< 1 \}$, its complement $\mathscr{E}^{c}$ and the associated indicator functions $\mathds{1}_{\mathscr{E}}, \mathds{1}_{\mathscr{E}^c}$ respectively, a generic term in the RHS of the above sum satisfies
\begin{align}
\mbox{tr}((I-\lambda_{a_{m},m})\rho_{m}^{\otimes n})\mathds{1}_{\mathscr{E}^{c}}
+ \mbox{tr}((I-\lambda_{a_{m},m})\rho_{m}^{\otimes n})\mathds{1}_{\mathscr{E}} \leq \mathds{1}_{\mathscr{E}^{c}}+\sum_{i=1}^{3}T_{2i}, \nonumber
\end{align}
where
\begin{align*}
T_{21} = 2\tr((I-\gamma_{a_{m},m})\rho_{m}^{\otimes n})\mathds{1}_{\mathscr{E}} ,\quad
T_{22}= 4\sum_{\hata \neq a_{m} }\tr(\gamma_{\hata,m}\rho_{m}^{\otimes n})\mathds{1}_{\mathscr{E}},\quad \text{and} \quad T_{23}=4\sum_{\hatm \neq m}\sum_{\tildea  }\tr(\gamma_{\tildea,\hatm}\rho_{m}^{\otimes n})\mathds{1}_{\mathscr{E}},
\end{align*}
where we have used Hayashi-Nagaoka inequality \cite{hayashi2003general}. 

\noindent \textbf{Distribution of the Random Code} : The objects $g_{I}\in \mathcal{V}^{k \times n},g_{O/I} \in \mathcal{V}^{l \times n},b^{n} \in \mathcal{V}^{n}$ and the collection $(a_{m} \in s(m): m \in \mathcal{V}^{l})$ specify an NCC CQ-PTP code unambiguously. A distribution for a random code is therefore specified through a distribution of these objects. We let upper case letters denote the associated random objects, and obtain
\begin{eqnarray}
 \label{Eqn:}
 \CalP\!\left(\begin{array}{c} G_{I}=g_{I},G_{O/I}=g_{O/I}\\B^{n}=b^{n},A_{m}=a_{m}: m \in S(m) \end{array}  \right)= q^{-(k+l+1)n}\prod_{m \in \CalV^{l}}\frac{1}{\Theta(m)},\nonumber
\end{eqnarray}
and analyze the expectation of $\mathscr{E}$ and the terms $T_{2i}; i\in [1,3]$ in regards to the above random code. 
We begin by $\mathbb{E}_{\mathcal{P}}[\mathscr{E}] = \mathcal{P}(\sum_{a \in \CalV^{k}} \mathds{1}_{\{V^{n}(a,m) \in T_{\delta}^{n}(p_{V}) \}} < 1). $ For this, we provide the following proposition.  
\begin{prop}\label{prop:PTP:Lemma for T1}
There exist  $\epsilon_{T_1}(\delta), \delta_{T_1}(\delta),$ such that for  all sufficiently small $\delta$ and sufficiently large $n$, we have $\EE_{\CalP}\left[\mathscr{E}\right] \leq \epsilon_{{T_1}}(\delta) $, if $ \frac{k}{n} \geq \log{q} - H(V) + \delta_S$, where  $\epsilon_{{S}},\delta_{S} \searrow 0$ as $\delta \searrow 0$.
\end{prop}
\begin{proof}
The proof follows from Appendix B of  \cite{MACwithStates_ArunPad_SandeepPra} with the identification of $\mathcal{S} = \phi $.
\end{proof}
We now consider $T_{21}$. Deriving an upper bound on $T_{21}$ is by deriving a lower bound $\tr(\lambda_{a_{m},m}\rho_{m}^{\otimes n})$. This follows by an argument that is colloquially referred to as `pinching'. Lemma \ref{Lem:CharHighProbSubAKAPinching} in Appendix \ref{AppSec:CharHighProbSubAKAPinching} proves the existence of $\lambda>0$ such that $\mathbb{E}_{\mathcal{P}}\{ T_{21}\} \leq \exp\{-n\lambda\delta^{2} \}$ for sufficiently large $n$.
\noindent We now analyze $\mathbb{E}_{\mathcal{P}}[T_{22}]$.
Denoting the event\begin{eqnarray}\label{def:JandK}
\mathcal{J}\define \left\{ \begin{array}{c}\Theta(m)\geq 1 ,\!V^{n}(\hat{a},\hatm)=\hatx^{n}\\A_{m}=d,V^{n}(d,m)=x^{n} \end{array} \right\}\subseteq\! \mathcal{K}\define \left\{  \begin{array}{c} V^{n}(\hat{a},\hatm)=\hatx^{n}\\V^{n}(d,m)=x^{n} \end{array}\right\}
\end{eqnarray} we perform the following steps.
\begin{align}
 \mathbb{E}_{\mathcal{P}}[T_{22}] = \sum_{\hata \in \CalV^{k}  }\mathbb{E}_{\mathcal{P}}[\mbox{tr}(\Gamma_{\hata,m}\rho_{m}^{\otimes n})\mathds{1}_{\{ \theta(m) \geq 1 \}}\mathds{1}_{\{\hata \neq A_{m}\}}] 
 &= \sum_{d \in \CalV^{k}}\sum_{\hata \in \CalV^{k} }\sum_{x^{n} \in T_{\delta}^{n}(p_{V})}\sum_{\hatx^{n} \in \CalV^{n}}\mathbb{E}\left[\mbox{tr}(\Gamma_{\hata,m}\rho_{m}^{\otimes n})\11_{\hata \neq d}\11_{\CalJ}\right]
 \nonumber\\
 &= \sum_{d \in \CalV^{k}}\sum_{\hata \neq d }\sum_{x^{n} \in T_{\delta}^{n}(p_{V})}\sum_{\hatx^{n} \in \CalV^{n}}\mathbb{E}\left[\mbox{tr}(\Gamma_{\hata,m}\rho_{m}^{\otimes n})\11_{\CalJ}\right] \nonumber
\end{align}
where the restriction of the summation $x^{n}$ to $T_{\delta}^{n}(p_{V})$ is valid since $S(m) > 1$ forces the choice $A_{m} \in S(m)$ such that $V^{n}(A_{m},m) \in T_{\delta}^{n}(p_{V})$. Going further, we have $ $
\begin{align}\label{eq:PTP_T22}
 \mathbb{E}_{\mathcal{P}}[T_{22}]& = \sum_{\substack{d,\hata \in \CalV^{k}\\ \hata \neq d}}\sum_{x^{n} \in T_{\delta}^{n}(p_{V})}\sum_{\hatx^{n} \in T_{\delta}^{n}(p_{V})}\!\!\!\!\!\!\mathbb{E}\left[\mbox{tr}(\pi_{\rho}\pi_{\hatx^{n}}\pi_{\rho}\rho_{x^{n}}^{\otimes n})\11_{\CalJ}\right]\nonumber \\
 &= \sum_{\substack{d,\hata: \hata \neq d}}\sum_{x^{n} \in T_{\delta}^{n}(p_{V})}\sum_{\hatx^{n} \in T_{\delta}^{n}(p_{V})}\mbox{tr}(\pi_{\rho}\pi_{\hatx^{n}}\pi_{\rho}\rho_{x^{n}}^{\otimes n})\mathcal{P}(\CalJ)\nonumber\\
 &\stackrel{(a)}{\leq}\sum_{\substack{d,\hata:\hata \neq d}}\sum_{\hat{x}^n  \in T_{\delta}^{n}(p_{V})}\mbox{tr}(\pi_{\hatx^{n}}\pi_{\rho} )\mathcal{P}(\CalJ)2^{-n\left[S(\rho)-H(p_{V})+\epsilon_V\right]} \nonumber\\
&\stackrel{(b)}{\leq} \sum_{\substack{d,\hata:\hata \neq d}}\sum_{{\hatx^{n} \in T_{\delta}^{n}(p_{V})}}\mbox{tr}(\pi_{\hatx^{n}}\pi_{\rho})\mathcal{P}(\mathcal{K})2^{-n\left[S(\rho)-H(p_{V})+\epsilon_V\right]} \nonumber\\
 &\stackrel{(c)}{=} \sum_{\substack{d,\hata:\hata \neq d}}\sum_{{\hatx^{n} \in T_{\delta}^{n}(p_{V})}}\mbox{tr}(\pi_{\hatx^{n}}\pi_{\rho})\frac{1}{q^{2n}}2^{-n\left[S(\rho)-H(p_{V})+\epsilon_V\right]}\nonumber\\
 &\stackrel{(d)}{\leq} 2^{-n\left[ \chi(\{p_{V};\rho_{v}\})+ \epsilon_V  -2H(p_{V}) -\frac{2k}{n}\log{q} + 2\log{q} \right]}, 
\end{align}
where the restriction of the summation $\hatx^{n}$ to $T_{\delta}^{n}(p_{V})$ follows from the fact that $\pi_{\hat{x}^{n}}$ is the zero projector if $\hatx^{n} \notin T_{\delta}^{n}(p_{V})$, (a) follows from the operator inequality $ \sum_{x^n \in T_{\delta}(p_{V})}\pi_{\rho}\rho_{x^{n}}\pi_{\rho} \leq 2^{n(H(p_{V})+\epsilon_V(\delta)) }\pi_{\rho}\rho^{\otimes n}\pi_{\rho} \leq 2^{n(H(p_{V})+\epsilon_V(\delta)-S(\rho)) }\pi_{\rho}$ found in \cite[Eqn. 20.34, 15.20]{2017BkOnline_Wil}, (b) follows from Def. \ref{def:JandK}, (c) follows from pairwise independence of the distinct codewords, and (d) follows from $\pi_{\rho} \leq I $ and \cite[Eqn. 15.77]{2017BkOnline_Wil} and $\epsilon_V(\delta) \searrow 0$ as $\delta \searrow 0$.
 We now derive an upper bound on $\mathbb{E}_{\mathcal{P}}[T_{23}]$. 
We have
\begin{align}
\mathbb{E}_{\mathcal{P}}[T_{23}]& {=} \sum_{d,\hata \in \CalV^{k}}\sum_{\hatm \neq m }\sum_{\substack{x^{n},\hatx^{n} \in \\T_{\delta}^{n}(p_{V})}}\!\!\!\mathbb{E}\!\left[\mbox{tr}(\pi_{\rho}\Pi_{\hata,\hatm}\pi_{\rho}\rho_{A_{m},m}^{\otimes n})\mathds{1}_{\mathcal{J}}\right] \nonumber\\
&{=} \sum_{d,\hata \in \CalV^{k}}\sum_{\hatm \neq m }\sum_{\substack{x^{n},\hatx^{n} \in T_{\delta}^{n}(p_{V})}}\mbox{tr}(\pi_{\hatx^{n}}\pi_{\rho}\rho_{x^{n}}^{\otimes n}\pi_{\rho})\mathcal{P}(\mathcal{J})\nonumber\\
 &{\leq} \sum_{d,\hata \in \CalV^{k}}\sum_{\hatm \neq m }\sum_{\substack{\hatx^{n} \in T_{\delta}^{n}(p_{V})}}\!\!\!\!\mbox{tr}(\pi_{\hatx^{n}}\pi_{\rho})\mathcal{P}(\mathcal{J})2^{-n\left[S(\rho)-H(p_{V})+\epsilon_V\right]}\nonumber\\
 &{\leq}  \sum_{d,\hata \in \CalV^{k}}\sum_{\hatm \neq m }\sum_{\substack{\hatx^{n} \in T_{\delta}^{n}(p_{V})}}\!\!\!\!\mbox{tr}(\pi_{\hatx^{n}}\pi_{\rho})\mathcal{P}(\mathcal{K})2^{-n\left[S(\rho)-H(p_{V})+\epsilon_V\right]}\nonumber\\
&{=} \sum_{d,\hata \in \CalV^{k}}\sum_{\hatm \neq m }\sum_{\substack{\hatx^{n} \in T_{\delta}^{n}(p_{V})}}\!\!\!\!\mbox{tr}(\pi_{\hatx^{n}}\pi_{\rho})\frac{1}{q^{2n}}2^{-n\left[S(\rho)-H(p_{V})+\epsilon_V\right]}\nonumber\\
 &{\leq} \;\; 2^{-n\left[ \chi(\{p_{V};\rho_{v}\})+ 2\log_{2}q-2H(p_{V}) -\frac{2k+l}{n}\log_{2}q +\epsilon_V \right]}, \nonumber
\end{align}
where the inequalities above uses similar reasoning as in \eqref{eq:PTP_T22}.

We have therefore obtained three bounds 
$\frac{k}{n} > 1-\frac{H(p_{V})}{\log_{2} q} $, $\frac{2k}{n} < 2+\frac{\chi(\{p_{V};\rho_{v}\})-2H(p_{V})}{\log_{2} q} $, $\frac{2k+l}{n} < 2+\frac{\chi(\{p_{V};\rho_{v}\})-2H(p_{V})}{\log_{2} q} $. A rate of $\chi(\{p_{V};\rho_{v}\})-\epsilon$ is achievable by choosing $\frac{k}{n} = 1-\frac{H(p_{V})}{\log_{2} q}+\frac{\epsilon}{2}$, $\frac{l}{n} = \frac{\chi(\{p_{V};\rho_{v}\})-\epsilon\log_2\sqrt{q}}{\log_{2} q}$ thus completing the proof.
\end{proof}
\section{Decoding Sum over CQ-MAC}
\label{sec:theorem2}
Throughout this section, the source alphabets $\mathcal{S}\define \mathcal{S}_{1}=\mathcal{S}_{2}=\mathcal{F}_{q}$ is a finite field with $q$ elements and the receiver intends to reconstruct the sum $f(S_{1},S_{2})=S_{1}\oplus_{q}S_{2}$ of the sources. As discussed in Sec.~\ref{Sec:CentralIdea}, we propose a `separation based' coding scheme consisting of a K\"orner Marton (KM) source code followed by a CQ MAC channel code designed to communicate the sum of the message indices input at the channel code encoders. The focus of this section is to design, analyze and thereby characterize performance of the latter CQ MAC channel code tasked to communicate the sum of messages. Towards that end, we begin with a definition.

\begin{definition}
 \label{Defn:CQMACChnlCode-MsgSumCap}
 Let $\mathcal{V}=\mathcal{F}_{q}$ be a finite field and $(\rho_{x_{1}x_{2}} \in \mathcal{D}(\mathcal{H}_{Y}): (x_{1},x_{2}) \in \mathcal{X}_{1}\times \mathcal{X}_{2})$ be a CQ-MAC. A CQ-MAC code $c_{m\oplus} = (n,\mathcal{I}_{1}=\mathcal{I}_{2}=\mathcal{F}_{q}^{l},e_{1},e_{2},\lambda_{[q^{l}]})$ of block-length $n$ for recovering $\mathcal{F}_{q}-$sum of messages consists of two encoders maps $e_{j} : \mathcal{V}^{l} \rightarrow \mathcal{X}_{j}^{n} : j \in [2]$, and a POVM $\lambda_{q^{l}}= \{ \lambda_{m} \in \mathcal{P}(\mathcal{H}_{Y}^{\otimes n}) : m \in \mathcal{V}^{l}\}$.
 
 An $\mathcal{F}_{q}-$\textit{message-sum} rate $R>0$ is achievable if given any sequence $l(n) \in \naturals : n \in \naturals$ such that $\limsup_{n \rightarrow \infty} \frac{l(n)\log q}{n}< R$, any sequence $p_{M_{1}M_{2}}^{(n)}$ of PMFs on $\mathcal{F}_{q}^{l(n)} \times \mathcal{F}_{q}^{l(n)}$, there exists a CQ-MAC code $c_{m\oplus}^{(n)} = (n,\mathcal{I}=\mathcal{I}_{1}=\mathcal{I}_{2}=\mathcal{F}_{q}^{l(n)},e_{1}^{(n)},e_{2}^{(n)},\lambda_{\mathcal{I}})$ of block-length $n$ for recovering $\mathcal{F}_{q}^{l(n)}-$sum of messages such that for every $\delta > 0$, have
\begin{eqnarray}
\label{Eqn:ProbStatementGeneralfrrorOfCode}
\lim_{n\rightarrow \infty}\overline{\xi}(c_{m\oplus}^{(n)}) &=& \lim_{n\rightarrow \infty}1-\!\!\!\! \sum_{\substack{(m_{1},m_{2}) \\\in \mathcal{I}_{1}\times \mathcal{I}_{2}}}\!\!\!\! p_{M_{1}}p_{M_{2}}(m_{1},m_{2})\tr(\lambda_{m_{1}\oplus m_{2}}\rho^{\otimes n}_{c,\ulinem}) =0
 \nonumber
\end{eqnarray}
where $\rho^{\otimes n}_{c,\ulinem} \define \otimes_{i=1}^{n}\rho_{x_{1i}(m_{1})x_{2i}(m_{2})}$, where $e_{j}(m_{j}) = x_{j1}(m_{j}),x_{j2}(m_{j}),\cdots, x_{jn}(m_{j})$ for $j\in [2]$. The closure of the set of all achievable $\mathcal{F}_{q}-$message-sum rates is the message-sum capacity of the CQ-MAC.
\end{definition}

From our discussion in Sec.~\ref{Sec:CentralIdea} and the above definition, a road map for characterizing sufficient conditions for computing the sum over a CQ-MAC must be evident. Referring back to Sec.~\ref{Sec:CentralIdea}, we note that is joint PMF $\mathbb{W}_{S_{1}S_{2}}$ of the sources is such that $H(S_{1}\oplus_{q}S_{2})$ is dominated by the message-sum capacity of the CQ-MAC, then the corresponding sum of sources can be reconstructed over the CQ-MAC. Therefore, if $R > 0$ is an achievable message-sum rate over a CQ-MAC, then $H(S_{1}\oplus_{q}S_{2}) < R$ is a sufficient condition. We now state the main contribution of this section - a lower bound on the message-sum capacity of a CQ-MAC. Following its proof, we leverage the above argument in Thm.~\ref{thm:SumCQ_MAC} to characterize sufficient conditions for reconstructing sum of sources over an arbitrary CQ-MAC.

\begin{definition}
 \label{Defn:AdmPMFsForMessageSumRate}
Given a CQ-MAC $\rho_{\ulineCalX} \define (\rho_{x_{1}x_{2}} \in \mathcal{D}(\mathcal{H}_{Y}): (x_{1},x_{2}) \in \mathcal{X}_{1}\times \mathcal{X}_{2})$ and a prime power $q$, let \begin{eqnarray}
\mathscr{P}(\rho_{\ulineCalX},q) \define \left\{\begin{array}{l}(p_{V_{1}V_{2}U},\rho_{u}:u \in \mathcal{V})\end{array}: \begin{array}{l}p_{V_{1}X_{1}}p_{V_{2}X_{2}} \mbox{ is a PMF on }
\mathcal{V}\times \mathcal{X}_{1} \times \mathcal{V}\times \mathcal{X}_{2}, \mathcal{V}=\mathcal{F}_{q},\\
\displaystyle p_{V_{1}V_{2}U}(v_{1},v_{2},u) =\sum_{x_{1},x_{2} \in \ulineCalX} p_{V_{1}X_{1}}(v_{1},x_{1})p_{V_{2}X_{2}}(v_{2},x_{2})\mathds{1}_{\{ u = v_{1}\oplus_{q} v_{2}\}}\\\displaystyle
\rho_{u} \deq \sum_{v_1\in\CalF_q}\sum_{v_2\in\CalF_q}\!\!{p_{V_1V_2|U}(v_1,v_2|u)}\rho_{v_1v_2}\11_{\{v_1 \oplus_q v_2 = u\}},\\ \displaystyle
    \rho_{v_1v_2} \deq \sum_{x_1\in \CalX_1,x_2\in\CalX_2} p_{X_1|V_1}(x_1|v_1)p_{X_2|V_2}(x_2|v_2)\rho_{x_1x_2}
\end{array}\right\}
\nonumber
\end{eqnarray}
For $(p_{V_{1}V_{2}U},\rho_{u}:u \in \mathcal{F}_{q}) \in \mathscr{P}(\rho_{\ulineCalX},q)$, let \begin{eqnarray}\mathscr{R}(p_{V_{1}V_{2}U},\rho_{\mathcal{V}}) \define \min\{ H(V_1), H(V_2) \} -  H(U) +\chi(\{p_U;\rho_{u}\})\mbox{ and }\mathscr{R}(\rho_{\ulineCalX},q) \define \displaystyle \sup_{p_{V_{1}V_{2}U},\rho_{\mathcal{V}} \in \mathscr{P}(\rho_{\ulineCalX},q)}\mathscr{R}(p_{V_{1}V_{2}U},\rho_{\mathcal{V}}). \end{eqnarray}
\end{definition}

\begin{lemma}
\label{Lem:DecSumOfMsgOnCQ-MAC}
 $\mathcal{F}_{q}-$message-sum rate $\mathscr{R}(\rho_{\ulineCalX},q)$ is achievable over a CQ-MAC $\rho_{\ulineCalX}= (\rho_{x_{1}x_{2}} \in \mathcal{D}(\mathcal{H}_{Y}): (x_{1},x_{2}) \in \mathcal{X}_{1}\times \mathcal{X}_{2})$. 
\end{lemma}
\begin{proof}
 Let $(p_{V_{1}V_{2}U},\rho_{u}:u \in \mathcal{V}) \in \mathscr{P}(\rho_{\ulineCalX},q)$ with associated collection $(\rho_{v_{1}v_{2}}: (v_{1},v_{2}) \in \mathcal{V}_{1}\times \mathcal{V}_{2})$ of density operators and PMF $p_{V_{1}X_{1}}p_{V_{2}X_{2}}$ on $\mathcal{V}_{1}\times \mathcal{X}_{1} \times \mathcal{V}_{2}\times \mathcal{X}_{2}$ where $\mathcal{V}_{1}=\mathcal{V}_{2}=\mathcal{F}_{q}$.
 
 We now describe the coding scheme in terms of a specific code. It is instructive to revisit Sec.~\ref{Sec:CentralIdea}, wherein we specified the import of both encoders employing cosets of the the \textit{same} linear code. In order to choose codewords of a desired empirical distribution $p_{V_{j}}$, we employ NCCs (as was done for the same reason in Sec.~\ref{Sec:NCCAchieveCQ-PTPCapacity}). Following the same notation as in proof of Thm.~\ref{Thm:NCCAchievesCQPTP}, we now specify the random coding scheme.
 
 Let $G_{I} \in \fieldq^{k \times n}, G_{O/I} \in \fieldq^{l \times n}, B_{j} \in \fieldq^{n} : j \in [2]$ be mutually independent and uniformly distributed on their respective range spaces. Through out this proof, we let $\oplus = \oplus_{q}$. Let $V_{j}^{n}(a,m_{j}) \define a {G_{I}} \oplus m_{j}G_{O/I}\oplus B_{j}^{n} : (a,m_{j})\in \fieldq^{k+l}$ for $j \in [2]$ and $U^{n}(a,m) \define a {G_{I}} \oplus m G_{O/I}\oplus B_{1}^{n}\oplus B_{2}^{n} : (a,m)\in \fieldq^{k+l}$. For $j \in [2]$, let
 \begin{eqnarray}
\label{Eqn:NoTypicalElements}
 S_{j}(m_{j}) \define \begin{cases} \{a \in \mathcal{V}^{k}: V^{n}_{j}(a,m_{j}) \in T_{\delta}^{n}(p_{V_{j}})\}&\mbox{if }\displaystyle \sum_{a \in \mathcal{V}^{k}}\mathds{1}_{\left\{ V_{j}^{n}(a,m_{j}) \in T_{\delta}^{n}(p_{V_{j}}) \right\}} \geq 1 \\\{0^{k}\}&\mbox{otherwise, i.e }\displaystyle \sum_{a \in \mathcal{V}^{k}}\mathds{1}_{\left\{ V_{j}^{n}(a,m_{j}) \in T_{\delta}^{n}(p_{V_{j}}) \right\}} = 0,\end{cases}
 \nonumber
\end{eqnarray}
for each $m_{j} \in \mathcal{V}^{l}$. For $m_{j} \in \mathcal{V}^{l}$, a predetermined element $A_{j,m_{j}} \in S_{j}(m_{j})$ is chosen. We let $\Theta_{j}(m_{j}) \define |S_{j}(m_{j})|$. For $m_{j} \in \mathcal{V}^{l}$, a predetermined $X_{j}^{n}(m_{j}) \in \mathcal{X}_{j}^{n}$ is chosen. As we shall see later, the choice of $X_{j}^{n}(m_{j})$ is based on $V_{j}^{n}(A_{j,m_{j}},m_{j})$. We are thus led to the encoding rule.
\med
\textit{Encoding Rule}: On receiving message $(m_{1},m_{2}) \in \mathcal{V}^{l} \times \mathcal{V}^{l}$, the quantum state $\rho_{m_{1}m_{2}} \define \rho_{X_{1}^{n} (m_{1})X_{2}^{n} (m_{2})} \define \otimes_{t=1}^{n}\rho_{X_{1t}(m_{1})X_{2t}(m_{2})}$ is (distributively) prepared.

\med\textit{Distribution of the Random Code}: The distribution of the random code is completely specified through the distribution $\mathcal{P}(\cdot)$ of $G_{I},G_{O/I}, B_{1}^{n},B_{2}^{n}, (A_{1,m_{1}}: m_{1} \in \CalV^{l}), (A_{2,m_{2}}: m_{2} \in \CalV^{l})$ and $(X_{j}^{n}(m_{j}) : m_{j} \in \mathcal{V}^{l})$. We let
\begin{eqnarray}
  \mathcal{P}\left( \!\!\!
  \begin{array}{c} (A_{1,m_{1}}=a_{1,m_{1}}: m_{1} \in \mathcal{V}^{l}), (A_{2,m_{2}}=a_{2,m_{2}}: m_{2} \in \mathcal{V}^{l}),\\B_{j}^{n} = b_{j}^{n} : j \in [2],
(X_{1}(m_{1}) = x_{1}^{n}(m_{1}): m_{1} \in \CalV^{l} ),\\ G_{I}=g_{I},G_{O/I}=g_{O/I},(X_{2}(m_{2}) = x_{2}^{n}(m_{2}): m_{2} \in \CalV^{l} ) \end{array} \!\!\!\right)= \displaystyle \left[\prod_{m_{1}}\frac{\mathds{1}_{\{a_{1,m_{1}} \in s_{1}(m_{1}) \}}}{\Theta(m_{1})}p_{X_{1}|V_{1}}^{n}(x_{1}^{n}(m_{1})|v_{1}^{n}(a_{1,m_{1}},m_{1}))\right]\times \nonumber\\
\label{Eqn:DistOfRandomCode}
\left[\prod_{m_{2}}\frac{\mathds{1}_{\{a_{2,m_{2}} \in s_{2}(m_{2})\} }}{\Theta(m_{2})}p_{X_{2}|V_{2}}^{n}(x_{2}^{n}(m_{2})|v_{2}^{n}(a_{2,m_{2}},m_{2})\right]
\times \frac{1}{q^{kn+ln+2n}}. \nonumber
\end{eqnarray}
Towards specifying a decoding POVM, we state the associated density operators modeling the quantum systems, their spectral decompositions and projectors. Let
\begin{eqnarray}
 \rho \deq 
 \sum_{y \in \mathcal{Y}}s_{Y}(y)\ket{h_{y}}\bra{h_{y}},~~ \rho_{x_{1}x_{2}}\deq \sum_{y \in \mathcal{Y}}p_{Y|X_{1}X_{2}}(y|x_{1},x_{2})\ket{e_{y|x_{1}x_{2}}}\bra{e_{y|x_{1}x_{2}}} : (x_{1},x_{2}) \in \ulineCalX \nonumber\\
 \rho_{v_{1}v_{2}} \deq  
 \sum_{y \in \mathcal{Y}}q_{Y|V_{1}V_{2}}(y|v_{1},v_{2})\ket{f_{y|v_{1}v_{2}}}\bra{f_{y|v_{1}v_{2}}}: (v_{1},v_{2}) \in \ulineCalV ,~~
 \rho_{u} \deq \sum_{y \in \mathcal{Y}}r_{Y|U}(y|u)\ket{g_{y|u}}\bra{g_{y|u}} : u \in \mathcal{U} \nonumber,
\end{eqnarray}

\med\textit{Decoding POVM}: Unlike a generic CQ-MAC decoder \cite{winter2001capacity}, which aims at decoding both the classical messages from the quantum state received, the decoder here is designed to decode only the sum of messages transmitted.
For this, the decoder employs the nested coset code $(n, k,l,G_I,G_{O/I},B^n)$, where $B^n = B_1^n \oplus B_2^n$. 
We define $U^n(a,m) \deq aG_I + mG_{O/I} + B^n$ to represent a generic codeword. We let $\Pi_{a,m} \deq \pi_{U^n(a,m)}\11_{\{U^n(a,m) \in \TDelta(p_U)\}}$, where $p_U$ is as defined in the theorem statement. The decoder is provided with a sub-POVM $\Lambda_{\CalI} \deq \{\Lambda_{m}\define\sum_{a \in \mathcal{F}_{q}^{k}}\Lambda_{a,m} : m \in \mathcal{F}_{q}^{l}\}$  where
\begin{align}
\Lambda_{a,m} &\deq \Big(\sum_{\hat{a} \in \CalF_q^k}\sum_{\hat{m}\in \CalF_q^l}\Gamma_{\hat{a},\hat{m}}\Big)^{-1/2}\Gamma_{a,m}    \Big(\sum_{\hat{a} \in \CalF_q^k}\sum_{\hat{m}\in \CalF_q^l}\Gamma_{\hat{a},\hat{m}}\Big)^{-1/2}, \nonumber
\end{align}
$\Lambda_{-1} \deq I - \sum_{{a} \in \CalF_q^k}\sum_{{m}\in \CalF_q^l} \Lambda_{a,m}$ and $\Gamma_{a,m} \deq   \pi_{\rho}\Pi_{(a,m)}\pi_{\rho}$. We note that
\begin{eqnarray} 
\displaystyle \pi_{\rho} \define \sum_{y^{n} \in T_{\delta}^{n}(s_{Y})}\displaystyle\bigotimes_{t=1}^{n}\ket{h_{y_{t}}}\bra{h_{y_{t}}}\mbox{ and }\pi_{u^n} \define \sum_{y^{n} : (u^{n},y^{n}) \in T_{\delta}^{n}(p_{U}r_{Y|U})}\displaystyle\bigotimes_{t=1}^{n}\ket{g_{y_{t}|u_{t}}}\bra{g_{y_{t}|u_{t}}}, \nonumber
\end{eqnarray}
denote the typical and conditional typical projectors (as stated in Definition 15.2.4 \cite{2013Bk_Wil}) with respect to $\rho \deq \sum_{u\in \CalF_q}p_U(u)\rho_u$ and $(\rho_u:u \in \mathcal{U})$, respectively. 

\med \textit{Error Analysis:}
We derive upper bounds on $\mathbb{E}_{\mathcal{P}}\{\overline{\xi}(c_{m\oplus})\}$. Our derivation will be similar to those adopted in proof of Thm.~\ref{Thm:NCCAchievesCQPTP}. Let us define event
\begin{align}
\label{Eqn:NotTypErrEvent}
 \mathscr{E} \define \left\{ \left(\!\!\! \begin{array}{c} V_{1}^{n}(A_{1.m_{1}},m_{1}), X_{1}^{n}(m_{1}),\\V_{2}^{n}(A_{2.m_{2}},m_{2}), X_{2}^{n}(m_{2}),\\ V_{1}^{n}(A_{1.m_{1}},m_{1}) \oplus V_{2}^{n}(A_{2.m_{2}},m_{2}) \end{array} \!\!\!\right) \in T_{{8\delta}}(p_{V_{1}X_{1}V_{2}X_{2}U})  \right\}.
\end{align}
We have
\begin{eqnarray}
 \label{Eqn:Msg-SumThmStep1PreHayaNaga1}
 \mathbb{E}_{\mathcal{P}}\left\{ \sum_{m_{1}}\sum_{m_{2}}p_{M_{1}M_{2}}(m_{1},m_{2})\tr(\left[I-\Lambda_{m_{1}\oplus m_{2}} \right])\rho_{m_{1}m_{2}}^{\otimes n}  \right\}\leq \underbrace{\mathbb{E}_{\mathcal{P}}\left\{ \sum_{m_{1}}\sum_{m_{2}}p_{M_{1}M_{2}}(m_{1},m_{2})\tr(\left[I-\Lambda_{m_{1}\oplus m_{2}} \right])\rho_{m_{1}m_{2}}^{\otimes n}  \mathds{1}_{\mathscr{E}^{c}}\right\}}_{T_{1}}\nonumber\\
 \label{Eqn:Msg-SumThmStep1PreHayaNaga2}
  + \underbrace{\mathbb{E}_{\mathcal{P}}\left\{ \sum_{m_{1}}\sum_{m_{2}}p_{M_{1}M_{2}}(m_{1},m_{2})\tr(\left[I-\Lambda_{m_{1}\oplus m_{2}} \right])\rho_{m_{1}m_{2}}^{\otimes n}  \mathds{1}_{\mathscr{E}}\right\}}_{T_{2}}.\nonumber
\end{eqnarray}
In regards to $T_{1}$, the sub-POVM nature of $\Lambda_{\mathcal{I}}$ and the fact that $\rho_{m_{1},m_{2}}^{\otimes n}$ is a density operator enables us conclude $T_{1} \leq \mathbb{E}_{\mathcal{P}}\{ \mathds{1}_{\mathscr{E}^{c}}\}$. Furthermore, observe that $X_{j}(m_{j)}$ is distributed with PMF $p_{X_{j}|V_{j}}^{n}$ conditionally on $V_{j}^{n}(A_{j,m_{j},m_{j}})$. (See (\ref{Eqn:DistOfRandomCode}). In addition, $p_{V_{1}X_{1}V_{2}X_{2}} = p_{V_{1}X_{1}}p_{V_{2}X_{2}}$ implies that, standard conditional typicality arguments yields
\begin{eqnarray}
 \label{Eqn:Msg-SumThmStep1-T1Anal1}
 \mathbb{E}_{\mathcal{P}}\{ \mathds{1}_{\mathscr{E}^{c}}\} \leq 
 \mathbb{E}_{\mathcal{P}}\left\{ \sum_{m_{1}}p_{M_{1}}(m_{1})\mathds{1}_{\{\Theta_{1}(m_{1}) = 0 \}}  + \sum_{m_{2}}p_{M_{2}}(m_{2})\mathds{1}_{\{\Theta_{1}(m_{2}) = 0 \}} \right\} + \exp\{ -n \delta\},
 \end{eqnarray}
where $\delta$ is chosen appropriately. In the above inequality, the second term on the RHS is an upper bound on the probability of the event $(X_{1}^{n}(m_{1},X_{2}^{n}(m_{2})) \notin T_{\delta}^{n}(p_{V_{1}X_{1}V_{2}X_{2}U}|v_{1}^{n},v_{2}^{n},v_{1}^{n}\oplus v_{2}^{n})$ conditioned on $ (V_{1}^{n}(A_{1.m_{1}},m_{1}),V_{2}^{n}(A_{2.m_{2}},m_{2}),V_{1}^{n}(A_{1.m_{1}},m_{1}) \oplus V_{2}^{n}(A_{2.m_{2}},m_{2})) =   (v_{1}^{n},v_{2}^{n},v_{1}^{n}\oplus v_{2}^{n}) \in T_{\delta}^{n}(p_{V_{1}V_{2}U})$ and the first term provides an upper bound on the complement of the latter event. An upper bound on $T_{1}$ therefore reduces to deriving an upper bound on the first term on the RHS of (\ref{Eqn:Msg-SumThmStep1-T1Anal1}). This task - deriving an upper bound on the first term on the RHS of (\ref{Eqn:Msg-SumThmStep1-T1Anal1}) - being a classical analysis, has been detailed in several earlier works \cite{201710TIT_PadPra, 201603TIT_PadSahPra, 201405PhDThe_Pad, 201804TIT_PadPra} and in particular \cite[Proof of Thm.~2.5]{2020Bk_PraPadShi} or \cite[Appendix B]{201710TIT_PadPra}. Following this, we have
\begin{eqnarray}
 \label{Eqn:Msg-SumThmStep1-T1Bnd1}
 \mathbb{E}_{\mathcal{P}}\left\{ \sum_{m_{j}}p_{M_{j}}(m_{j})\mathds{1}_{\{\Theta_{j}(m_{j}) = 0 \}}\right\} &\leq& \exp\left\{ -n \left( \frac{k\log q}{n} - [\log q- H(V_{j})] \right) \right\} 
\end{eqnarray}
 thereby ensurnig $T_{1} \leq 2\exp\{-n\delta\}$ if
\begin{align}
\frac{k\log q}{n} &\geq  \max\left\{ \log q- H(V_{1}), \log q- H(V_{2})\right\} = \log q - \min\{H(V_{1}),H(V_{2}) \} . \label{Eqn:Msg-SumThmStep1-T1Bnd2}
\end{align}
We now analyze $T_{2}$. Applying the Hayashi-Nagaoka inequality, we have$T_{2} \stackrel{(a)}{\leq} T_{21}+T_{22}+T_{23},$ where
\begin{eqnarray}
 \label{Eqn:Msg-SumThmStep1-T2-1}
 T_{21}\deq \mathbb{E}_{\mathcal{P}}\left\{ 2 \sum_{m_{1}}\sum_{m_{2}}p_{M_{1}M_{2}}(m_{1},m_{2}) \tr([I-\Gamma_{A_{\ulinem}^{\oplus},\ulinem^{\oplus}}]\rho_{m_{1}m_{2}}^{\otimes n}  ])\mathds{1}_{\mathscr{E}}\right\} \\ 
 T_{22} \define \mathbb{E}_{\mathcal{P}}\left\{4 \sum_{m_{1}}\sum_{m_{2}}\sum_{\hata \neq A_{\ulinem}^{\oplus}}p_{M_{1}M_{2}}(m_{1},m_{2}) \tr(\Gamma_{\hata,\ulinem^{\oplus}}\rho_{m_{1}m_{2}}^{\otimes n}  )\mathds{1}_{\mathscr{E}}\right\}\nonumber\\
 T_{23} \define \mathbb{E}_{\mathcal{P}}\left\{4 \sum_{m_{1}}\sum_{m_{2}}\sum_{\hata \neq A_{\ulinem}^{\oplus}}\sum_{\hatm \neq \ulinem^{\oplus}}p_{M_{1}M_{2}}(m_{1},m_{2}) \tr(\Gamma_{\hata,\hatm}\rho_{m_{1}m_{2}}^{\otimes n}  )\mathds{1}_{\mathscr{E}}\right\},\nonumber
\end{eqnarray}
and  $A_{\ulinem}^{\oplus}\define A_{1,m_{1}} \oplus A_{2,m_{2}} \in \CalV^{k}, \ulinem^{\oplus} \define m_{1}\oplus m_{2}  \in \CalV^{l}$. We note that (\ref{Eqn:Msg-SumThmStep1-T2-1}(a)) follows from an argument analogous to the one in (\ref{Eqn:CQ-PTPErrorProbability}). We now analyze $T_{21},T_{22}$ and $T_{23}$. We begin with $T_{21}$. Deriving an upper bound on $T_{21}$ is by deriving a lower bound $\mathbb{E}_{\mathcal{P}}\left\{ \tr(\Gamma_{A_{\ulinem}^{\oplus},\ulinem^{\oplus}}\rho_{m_{1}m_{2}}^{\otimes n}  ])\mathds{1}_{\mathscr{E}} \right\}$. This follows by an argument that is colloquially referred to as `pinching'. Refer to Lemma \ref{Lem:CharHighProbSubAKAPinching} in Appendix \ref{AppSec:CharHighProbSubAKAPinching}. Set $\mathcal{A}=\mathcal{V}=\mathcal{F}_{q}$, $\mathcal{B} = \mathcal{X}$, $p_{AB}=p_{V_{1}\oplus V_{2},X}$ and the density operators correspondingly. With this choice, Lemma \ref{Lem:CharHighProbSubAKAPinching} proves the existence of $\lambda>0$ such that $\mathbb{E}_{\mathcal{P}}\left\{ \tr(\Gamma_{A_{\ulinem}^{\oplus},\ulinem^{\oplus}}\rho_{m_{1}m_{2}}^{\otimes n}  ])\mathds{1}_{\mathscr{E}} \right\} \geq 1- \exp\{-n\lambda\delta^{2} \}$ for sufficiently large $n$.

We now analyze $T_{22}$. Denoting the event
\begin{eqnarray}\label{Eqn:Msg-SumThmStep1-T22-1}
\mathcal{J} \!\define\! \left\{ \left(\!\!\! \begin{array}{c} V_{1}^{n}(A_{1.m_{1}},m_{1}), X_{1}^{n}(m_{1}),V_{2}^{n}(A_{2.m_{2}},m_{2}), X_{2}^{n}(m_{2}) \end{array} \!\!\!\right)= (v_{1}^{n},x_{1}^{n},v_{2},x_{2}) \in T_{\delta_{4}}(p_{V_{1}X_{1}V_{2}X_{2}})  \right\},
\end{eqnarray}
 abbreviating $\ulinev_{\oplus}^{n}=v_{1}^{n}\oplus v_{2}^{n} $, $\ulinea^{\oplus}=a_{1}\oplus a_{2}$, we have
\begin{align}
\label{Eqn:Msg-SumThmStep1-T22-2}
 \mathbb{E}_{\mathcal{P}}[T_{22}] &= \mathbb{E}_{\mathcal{P}}\left\{4 \sum_{\ulinem}\sum_{a_{1},a_{2}}\sum_{\hata \neq \ulinea^{\oplus}}\sum_{\substack{(\ulinev^{n},\ulinex) \in \\ T_{\delta_{4}}(p_{\ulineV\ulineX})}}p_{\ulineM}(\ulinem) \tr(\Gamma_{\hata,\ulinem^{\oplus}}\rho_{m_{1}m_{2}}^{\otimes n}  )\mathds{1}_{\mathcal{J}}\mathds{1}_{\left\{ A_{j,m_{j}}=a_{j}:j \in [2] \right\}}\right\} \\
 \label{Eqn:Msg-SumThmStep1-T22-3}
 &\!\!\!\!\!\!\!\!\!\!\!\!= 4 \sum_{\ulinem}\sum_{a_{1},a_{2}}\sum_{\hata \neq \ulinea^{\oplus}}\sum_{\substack{(\ulinev^{n},\ulinex) \in \\ T_{\delta_{4}}(p_{\ulineV\ulineX})}}\sum_{\hatv^{n} \in \mathcal{V}^{n}}p_{\ulineM}(\ulinem) \tr(\pi_{\hatv^{n}}\pi_{\rho}\rho_{x_{1}^{n}x_{2}^{n}}^{\otimes n}\pi_{\rho}  )\mathbb{E}_{\mathcal{P}}
 \left\{\mathds{1}_{\mathcal{J}}\mathds{1}_{\left\{\!\!\!\begin{array}{c} A_{j,m_{j}}=a_{j}:j \in [2]\\ U^{n}(\hata,m_{1}\oplus m_{2})=\hatv^{n}\end{array}\!\!\!\right\}}\right\} \\
 \label{Eqn:Msg-SumThmStep1-T22-4}
 &\!\!\!\!\!\!\!\!\!\!\!\!\leq 4 \sum_{\ulinem}\sum_{a_{1},a_{2}}\sum_{\hata \neq \ulinea^{\oplus}}\sum_{\substack{(\ulinev^{n}) \in\\ T_{\delta_{4}}(p_{\ulineV})}}\sum_{\ulinex^{n}\in \ulineCalX^{n}}\sum_{\hatv^{n} \in \mathcal{V}^{n}}p_{\ulineM}(\ulinem)\left[\prod_{j=1}^{2}p_{X_{j}|V_{j}}(x_{j}^{n}|v_{j}^{n})\right] \tr(\pi_{\hatv^{n}}\pi_{\rho}\rho_{x_{1}^{n}x_{2}^{n}}^{\otimes n}\pi_{\rho}  )\mathcal{P}\left(\!\!\!
 { \begin{array}{c} V_{j}^{n}(a_{j},m_{j})=v_{j}^{n}\\ A_{j,m_{j}}=a_{j}:j \in [2]\\ U^{n}(\hata,m_{1}\oplus m_{2})=\hatv^{n}\end{array}\!\!\!}\right) 
 \end{align}
 \begin{align}\label{Eqn:Msg-SumThmStep1-T22-5}
 &\!\!\!\!\!\!\!\!\!\!\!\!= 4 \sum_{\ulinem}\sum_{a_{1},a_{2}}\sum_{\hata \neq \ulinea^{\oplus}}\sum_{\substack{(\ulinev^{n}) \in\\ T_{\delta_{4}}(p_{\ulineV})}}\sum_{\hatv^{n} \in \mathcal{V}^{n}}p_{\ulineM}(\ulinem) \tr(\pi_{\hatv^{n}}\pi_{\rho}\rho_{v_{1}^{n}v_{2}^{n}}^{\otimes n}\pi_{\rho}  )\mathcal{P}\left(\!\!\!
 { \begin{array}{c} V_{1}^{n}(a_{1},m_{1})=v_{1}^{n}, A_{j,m_{j}}=a_{j}:j \in [2]\\ V_{2}^{n}(a_{2},m_{2})=v_{2}^{n}, U^{n}(\hata,m_{1}\oplus m_{2})=\hatv^{n}\end{array}\!\!\!}\right)\\
 \label{Eqn:Msg-SumThmStep1-T22-6}
 &\!\!\!\!\!\!\!\!\!\!\!\!\leq 4 \sum_{\ulinem}\sum_{a_{1},a_{2}}\sum_{\hata \neq \ulinea^{\oplus}}\sum_{\substack{(\ulinev^{n}) \in\\ T_{\delta_{4}}(p_{\ulineV})}}\sum_{\hatv^{n} \in \mathcal{V}^{n}}p_{\ulineM}(\ulinem) \tr(\pi_{\hatv^{n}}\pi_{\rho}\rho_{v_{1}^{n}v_{2}^{n}}^{\otimes n}\pi_{\rho}  )\mathcal{P}\left(\!\!\!
 { \begin{array}{c} V_{1}^{n}(a_{1},m_{1})=v_{1}^{n}, \\ V_{2}^{n}(a_{2},m_{2})=v_{2}^{n}, U^{n}(\hata,m_{1}\oplus m_{2})=\hatv^{n}\end{array}\!\!\!}\right)\\
 \label{Eqn:Msg-SumThmStep1-T22-7}
 &\!\!\!\!\!\!\!\!\!\!\!\!\leq 4 \sum_{\ulinem}\sum_{a_{1},a_{2}}\sum_{\hata \neq \ulinea^{\oplus}}\sum_{\substack{(\ulinev^{n}) \in\\ T_{\delta_{4}}(p_{\ulineV})}}\sum_{\hatv^{n} \in \mathcal{V}^{n}}\!\!\!\!p_{\ulineM}(\ulinem) \tr(\pi_{\hatv^{n}}\pi_{\rho}\rho_{v_{1}^{n}v_{2}^{n}}^{\otimes n}\pi_{\rho}  )\frac{1}{q^{3n}} \nonumber\\ 
 &\!\!\!\!\!\!\!\!\!\!\!\! \leq 4 \sum_{\substack{ \ulinem,\\ a_{1},a_{2}}}\sum_{\hata }\sum_{\hatv^{n} \in \mathcal{V}^{n}}\!\!\!\!p_{\ulineM}(\ulinem) \tr(\pi_{\hatv^{n}}\pi_{\rho}  )\frac{2^{-n(S(\rho)-H(V_{1},V_{2})-4\delta )}}{q^{3n}} \\
 \label{Eqn:Msg-SumThmStep1-T22-8}
 &\!\!\!\!\!\!\!\!\!\!\!\!=4 \sum_{\substack{ \ulinem,\\ a_{1},a_{2}}}\sum_{\hata }\sum_{\hatv^{n} \in T_{\delta}(V_{1}\oplus V_{2})}p_{\ulineM}(\ulinem) \tr(\pi_{\hatv^{n}}\pi_{\rho}  )\frac{2^{-n(S(\rho)-H(V_{1},V_{2})-4\delta )}}{q^{3n}} \\
 \label{Eqn:Msg-SumThmStep1-T22-9}
 &\!\!\!\!\!\!\!\!\!\!\!\!\leq 4 \sum_{\substack{ \ulinem,\\ a_{1},a_{2}}}\sum_{\hata }p_{\ulineM}(\ulinem) \frac{\exp{-n\left[ S(\rho)-H(V_{1},V_{2})-4\delta-\displaystyle\sum_{u}p_{V_{1}\oplus V_{2}}(u)S(\rho_{u}) -H(V_{1}\oplus V_{2})\right] }}{q^{3n}}\\
 \label{Eqn:Msg-SumThmStep1-T22-10}
 &\!\!\!\!\!\!\!\!\leq 4\exp{-n\left\{ \left[\log q - \left(\displaystyle\sum_{u}p_{V_{1}\oplus V_{2}}(u)S(\rho_{u}) +H(V_{1}\oplus V_{2})- S(\rho)\right)\right]+\left[2\log q -H(V_{1},V_{2})\right]-\frac{3k\log q}{n}  \right\}}
\end{align}
where (i) (\ref{Eqn:Msg-SumThmStep1-T22-3}) follows from a summing over possible choices for $U^{n}(\hata,m_{1}\oplus m_{2})$, (ii) (\ref{Eqn:Msg-SumThmStep1-T22-4}) follows from evaluating expectation, enlarging the summation range of $x_{1}^{n},x_{2}^{n}$ and substituting the distribution of the random code, (iii) (\ref{Eqn:Msg-SumThmStep1-T22-5}) follows from the definitions of $\rho_{v_{1}v_{2}} : \ulinev \in \ulineCalV$, (iv) (\ref{Eqn:Msg-SumThmStep1-T22-6}) follows as an upper bound since the event in question has been enlarged, (v) (\ref{Eqn:Msg-SumThmStep1-T22-7}) follows from \cite[Lemma N.0.21\(c\)]{201405PhDThe_Pad} and the operator inequality $ \sum_{x^n \in T_{\delta}(p_{V})}\pi_{\rho}\rho_{x^{n}}\pi_{\rho} \leq 2^{n(H(p_{V})+\epsilon_V(\delta)) }\pi_{\rho}\rho^{\otimes n}\pi_{\rho} \leq 2^{n(H(p_{V})+\epsilon_V(\delta)-S(\rho)) }\pi_{\rho}$ found in \cite[Eqn. 20.34, 15.20]{2017BkOnline_Wil}, (vi) (\ref{Eqn:Msg-SumThmStep1-T22-8}) follows from the definition of $\pi_{\hatv^{n}}$ which is the $0$ projector if $\hatv^{n}$ is not typical wrt $p_{V_{1}\oplus V_{2}}$, (vii) (\ref{Eqn:Msg-SumThmStep1-T22-9}) follows from $\pi_{\rho} \leq I $ and \cite[Eqn. 15.77]{2017BkOnline_Wil} and, (viii) (\ref{Eqn:Msg-SumThmStep1-T22-10}) collating all the bounds. We now analyze $T_{23}$. 

\begin{eqnarray}
\label{Eqn:Msg-SumThmStep1-T23-2}
 \lefteqn{\mathbb{E}_{\mathcal{P}}[T_{23}] = \mathbb{E}_{\mathcal{P}}\left\{4 \sum_{\substack{\ulinem\in \mathcal{V}^{2l}\\\hatm \neq m_{1}\oplus m_{2} }}\sum_{\substack{\ulinea \in \CalV^{2k} \\\hata \in \CalV^{k}}}\sum_{\substack{(\ulinev^{n},\ulinex) \in \\ T_{\delta_{4}}(p_{\ulineV\ulineX})}}p_{\ulineM}(\ulinem) \tr(\Gamma_{\hata,\ulinem^{\oplus}}\rho_{m_{1}m_{2}}^{\otimes n}  )\mathds{1}_{\mathcal{J}}\mathds{1}_{\left\{ A_{j,m_{j}}=a_{j}:j \in [2] \right\}}\right\}}  \\
 \label{Eqn:Msg-SumThmStep1-T23-3}
 &&\!\!\!\!\!\!\!\!\!\!\!\!= 4\!\!\!\! \sum_{\substack{\ulinem\in \mathcal{V}^{2l}\\\hatm \neq m_{1}\oplus m_{2} }}\sum_{\substack{\ulinea \in \CalV^{2k} \\\hata \in \CalV^{k}}}\sum_{\substack{(\ulinev^{n},\ulinex) \in \\ T_{\delta_{4}}(p_{\ulineV\ulineX})}}\sum_{\hatv^{n} \in \mathcal{V}^{n}}p_{\ulineM}(\ulinem) \tr(\pi_{\hatv^{n}}\pi_{\rho}\rho_{x_{1}^{n}x_{2}^{n}}^{\otimes n}\pi_{\rho}  )\mathbb{E}_{\mathcal{P}}
 \left\{\mathds{1}_{\mathcal{J}}\mathds{1}_{\left\{\!\!\!\begin{array}{c} A_{j,m_{j}}=a_{j}:j \in [2]\\ U^{n}(\hata,\hatm)=\hatv^{n}\end{array}\!\!\!\right\}}\right\} \\
 \label{Eqn:Msg-SumThmStep1-T23-4}
 &&\!\!\!\!\!\!\!\!\!\!\!\!\leq 4 \!\!\!\!\sum_{\substack{\ulinem\in \mathcal{V}^{2l}\\\hatm \neq m_{1}\oplus m_{2} }}\sum_{\substack{\ulinea \in \CalV^{2k} \\\hata \in \CalV^{k}}}\sum_{\substack{(\ulinev^{n}) \in\\ T_{\delta_{4}}(p_{\ulineV})}}\sum_{\ulinex^{n}\in \ulineCalX^{n}}\sum_{\hatv^{n} \in \mathcal{V}^{n}}p_{\ulineM}(\ulinem)\left[\prod_{j=1}^{2}p_{X_{j}|V_{j}}(x_{j}^{n}|v_{j}^{n})\right] \tr(\pi_{\hatv^{n}}\pi_{\rho}\rho_{x_{1}^{n}x_{2}^{n}}^{\otimes n}\pi_{\rho}  )\mathcal{P}\left(\!\!\!
 { \begin{array}{c} V_{j}^{n}(a_{j},m_{j})=v_{j}^{n}\\ A_{j,m_{j}}=a_{j}:j \in [2]\\ U^{n}(\hata,\hatm)=\hatv^{n}\end{array}\!\!\!}\right)\\
 \label{Eqn:Msg-SumThmStep1-T23-5}
 &&\!\!\!\!\!\!\!\!\!\!\!\!= 4 \!\!\!\!\sum_{\substack{\ulinem\in \mathcal{V}^{2l}\\\hatm \neq m_{1}\oplus m_{2} }}\sum_{\substack{\ulinea \in \CalV^{2k} \\\hata \in \CalV^{k}}}\sum_{\substack{(\ulinev^{n}) \in\\ T_{\delta_{4}}(p_{\ulineV})}}\sum_{\hatv^{n} \in \mathcal{V}^{n}}p_{\ulineM}(\ulinem) \tr(\pi_{\hatv^{n}}\pi_{\rho}\rho_{v_{1}^{n}v_{2}^{n}}^{\otimes n}\pi_{\rho}  )\mathcal{P}\left(\!\!\!
 { \begin{array}{c} V_{1}^{n}(a_{1},m_{1})=v_{1}^{n}, A_{j,m_{j}}=a_{j}:j \in [2]\\ V_{2}^{n}(a_{2},m_{2})=v_{2}^{n}, U^{n}(\hata,\hatm)=\hatv^{n}\end{array}\!\!\!}\right)\\
 \label{Eqn:Msg-SumThmStep1-T23-6}
 &&\!\!\!\!\!\!\!\!\!\!\!\!\leq 4 \!\!\!\!\sum_{\substack{\ulinem\in \mathcal{V}^{2l}\\\hatm \neq m_{1}\oplus m_{2} }}\sum_{\substack{\ulinea \in \CalV^{2k} \\\hata \in \CalV^{k}}}\sum_{\substack{(\ulinev^{n}) \in\\ T_{\delta_{4}}(p_{\ulineV})}}\sum_{\hatv^{n} \in \mathcal{V}^{n}}p_{\ulineM}(\ulinem) \tr(\pi_{\hatv^{n}}\pi_{\rho}\rho_{v_{1}^{n}v_{2}^{n}}^{\otimes n}\pi_{\rho}  )\mathcal{P}\left(\!\!\!
 { \begin{array}{c} V_{1}^{n}(a_{1},m_{1})=v_{1}^{n}, \\ V_{2}^{n}(a_{2},m_{2})=v_{2}^{n}, U^{n}(\hata,\hatm)=\hatv^{n}\end{array}\!\!\!}\right)\\
 \label{Eqn:Msg-SumThmStep1-T23-7}
 &&\!\!\!\!\!\!\!\!\!\!\!\!\leq 4 \!\!\!\!\sum_{\substack{\ulinem\in \mathcal{V}^{2l}\\\hatm \neq m_{1}\oplus m_{2} }}\sum_{\substack{\ulinea \in \CalV^{2k} \\\hata \in \CalV^{k}}}\sum_{\substack{(\ulinev^{n}) \in\\ T_{\delta_{4}}(p_{\ulineV})}}\sum_{\hatv^{n} \in \mathcal{V}^{n}}\!\!\!\!p_{\ulineM}(\ulinem) \tr(\pi_{\hatv^{n}}\pi_{\rho}\rho_{v_{1}^{n}v_{2}^{n}}^{\otimes n}\pi_{\rho}  )\frac{1}{q^{3n}} \leq 4 \!\!\!\!\sum_{\substack{\ulinem\in \mathcal{V}^{2l}\\\hatm \neq m_{1}\oplus m_{2} }}\sum_{\substack{\ulinea \in \CalV^{2k} \\\hata \in \CalV^{k}}}\sum_{\hatv^{n} \in \mathcal{V}^{n}}\!\!\!\!p_{\ulineM}(\ulinem) \tr(\pi_{\hatv^{n}}\pi_{\rho}  )\frac{2^{-n(S(\rho)-H(V_{1},V_{2})-4\delta )}}{q^{3n}} \nonumber
\end{eqnarray}
 \begin{eqnarray}
 \label{Eqn:Msg-SumThmStep1-T23-8}
 &&\!\!\!\!\!\!\!\!\!\!\!\!=4 \sum_{\substack{ \ulinem,\\ a_{1},a_{2}}}\sum_{\hata }\sum_{\hatv^{n} \in T_{\delta}(V_{1}\oplus V_{2})}p_{\ulineM}(\ulinem) \tr(\pi_{\hatv^{n}}\pi_{\rho}  )\frac{2^{-n(S(\rho)-H(V_{1},V_{2})-4\delta )}}{q^{3n}} \\
 \label{Eqn:Msg-SumThmStep1-T23-9}
 &&\!\!\!\!\!\!\!\!\!\!\!\!\leq 4 \sum_{\substack{ \ulinem,\\ a_{1},a_{2}}}\sum_{\hata }p_{\ulineM}(\ulinem) \frac{\exp{-n\left[ S(\rho)-H(V_{1},V_{2})-4\delta-\displaystyle\sum_{u}p_{V_{1}\oplus V_{2}}(u)S(\rho_{u}) -H(V_{1}\oplus V_{2})\right] }}{q^{3n}}\\
 \label{Eqn:Msg-SumThmStep1-T23-10}
 &&\!\!\!\!\!\!\!\!\leq 4\exp{-n\left\{ \left[\log q - \left(\displaystyle\sum_{u}p_{V_{1}\oplus V_{2}}(u)S(\rho_{u}) +H(V_{1}\oplus V_{2})- S(\rho)\right)\right]+\left[2\log q -H(V_{1},V_{2})\right]-\frac{(3k+l)\log q}{n}  \right\}}.
\end{eqnarray}
The above sequence of steps is analogous to those used to derive an upper bound on $T_{22}$ and follow from the same set of arguments as provided for the bounds (\ref{Eqn:Msg-SumThmStep1-T22-3}) - (\ref{Eqn:Msg-SumThmStep1-T22-10}). This completes the proof of the claimed statement.
\end{proof}

We conclude this section with our main result in regards to decoding sum of sources. The proof of the following theorem follows from the discussion provided just prior to Definition \ref{Defn:AdmPMFsForMessageSumRate}. We therefore omit a detailed proof but just state the encoding and decoding techniques for completeness.

\begin{theorem}\label{thm:SumCQ_MAC}
The sum of a pair of sources distributed with PMF $\WW_{S_1S_2}$ can be reconstructed on a CQ-MAC $\rho_{\ulineCalX}= (\rho_{x_{1}x_{2}} \in \mathcal{D}(\mathcal{H}_{\ulineY}): (x_{1},x_{2}) \in \mathcal{X}_{1}\cross \mathcal{X}_{2})$ if $ H(S_1 \oplus_q S_2) < \mathscr{R}(\rho_{\ulineCalX},q)$.
\end{theorem}

\begin{proof}
We begin with an outline of our coding scheme. As stated in Sec.~\ref{Sec:CentralIdea}, we propose a `separation based approach' with two modules - source and channel. The source coding module employs a (distributed) K\"orner Marton (KM) source code. Specifically, \cite{197903TIT_KorMar} guarantees the existence of a parity check matrix $h \in \mathcal{F}_{q}^{ l \times n}= \mathcal{S}^{l \times n}$ and a decoder map $d:\mathcal{F}_{q}^{l} \rightarrow \mathcal{S}^{n}$ such that $\sum_{\ulines^{n} \in \ulineCalS^{n}} \mathbb{W}_{\ulineS}^{n} (\ulines^{n})\mathds{1}_{\{d(hs_{1}^{n} \oplus_{q} hs_{2}^{n}) \neq s_{1}^{n}\oplus_{q} s_{2}^{n}\}} \leq \epsilon$,
for any $\epsilon > 0$, and sufficiently large $n$, so long as $ \frac{l\log_{2}q}{n} > H(S_{1}\oplus_{q} S_{2})$. 

Both encoders of this KM source coding module employ one such parity check matrix $h \in \mathcal{F}_{q}^{ l \times n}$. The decoder of the KM source code employs the corresponding decoder map $d$. KM Source encoder $j$ outputs $M_{j}^{l} = h(S_{j}^{n})$. If the KM source decoder is provided $M_{1}^{l}\oplus_{q} M_{2}^{l}$, then it can reconstruct $S_{1}^{n}\oplus_{q}S_{2}^{n}$ with reliability at least $1-\epsilon$. The task of the CQ-MAC channel coding module is to make $M_{1}^{l}\oplus_{q} M_{2}^{l}$ available to the KM source decoder. We are thus confronted with the problem of designing a CQ-MAC channel coding module that can reliably communicate the sum of the messages indices that are input at the encoders.

Specifically, this channel coding module must communicate $M_{1}^{l}\oplus_{q} M_{2}^{l} \in \mathcal{F}_{q}^{l}$ within $n$ channel uses. If we can prove that there exists a MAC channel coding module for sufficiently large $n$ so long as
\begin{eqnarray}
\frac{l\log_{2}q}{n} < \min\{\!H(V_1),\!H(V_2)\!\}\! - \! H(U)\! +\!\chi(\{p_U;\rho_{u}\})
\nonumber
\end{eqnarray}
for any choice of auxiliary 

The source module employs the KM code. The corresponding KM decoder at the receiver only requires the sum of the message indices output by the KM code. The CQ-MAC channel coding module needs to communicate only the \textit{sum} of the two message indices input by the two encoders. 
Given $\epsilon_c$, we seek to identify a CQ-MAC code $c = (n,e_1,e_2,\CalM)$ such that $\overline{\xi}(c) \leq \epsilon_c$.

\noindent \textbf{Encoding:} The process of mapping source sequences to the CQ-MAC channel inputs is divided into two stages. In the first stage, a distributed source code proposed in \cite{197903TIT_KorMar} is employed which maps the $n-$length source sequences to message indices taking values over $\CalF_q^l$. For the  second stage we develop functions mapping these message indices to channel input codewords. We begin by defining the first stage of encoding which relies on Lemma 1 of \cite{201301arXivComputation_PadPra}.
This lemma guarantees the existence of a parity check matrix $h \in \CalF_q^{l(n)\cross n}$ and a map $d:\CalF_q^l(n) \rightarrow \CalF_q^n$,  for a sufficiently large $n$, such that (i) $\frac{l(n)}{n} \leq H(S_1 \oplus_q S_2) + \frac{\epsilon_c}{2}$ and (ii) $\PP(d(hS_1^n \oplus hS_2^n) \neq S_1^n \oplus S_2^n) \leq \frac{\epsilon_c}{2}$. We use one such parity matrix which satisfies the above conditions and define $M_j \deq hS_j,$ for $j=1,2$. This forms our first stage encoder.

Moving on to the second stage encoder, let us denote the maps of the two encoders as $\kappa_j: \CalF_q^l \rightarrow \CalX_j^n: j=1,2.$. For this stage, we use the NCC encoding developed in Section \ref{Sec:NCCAchieveCQ-PTPCapacity} for a CQ-PTP. 
Consider two NCCs with parameters $(n,k,l,g_I,g_{O/I},b_j^n): j\in \{1,2\}$  with range spaces as $ v_j^n(a,m_j) \deq ag_I \oplus m_j g_{O/I} \oplus b_j^n :j\in \{1,2\}, $ respectively. Note that the two NCCs share the common $g_{I}$ and $g_{O/I},$ but not  necessarily the bias vector $b_j^n$. Encoder  $j$ then constructs its NCC CQ-PTP code $(n,\CalI=\CalF_q^l,e_j,\lambda^j_{\CalI})$ using the corresponding NCC $(n,k,l,g_I,g_{O/I},b_j^n)$ as described in Definition \ref{Eqn:CQPTPPOVMDefn}. This defines the second stage encoding. Integrating the two stages, we obtain the following. To transmit the source sequence pair $(s_1^n,s_2^n)$ the sequence pair $(e_1(hs_1^n), e_2(hs_2^n))$ is send over the CQ-MAC channel which produces the quantum state $\rho_{e_1(hs_1^n), e_2(hs_2^n)}$ as the output. 


After performing the measurement and decoding the message $\hat{m}$, the decoder then employs the KM decoder $d(.)$ to obtain the sum of sources $d(\hat{m})$. An analysis of this coding scheme is provided in the Proof of Lemma \ref{Lem:DecSumOfMsgOnCQ-MAC}.



\end{proof}
\section{Decoding arbitrary functions over CQ-MAC}
\label{sec:arbitraryCQ-MAC}

Leveraging the technique developed in Theorem \ref{thm:SumCQ_MAC}, we provide the following proposition to reconstruct an arbitrary function of the sources
\begin{prop}
\label{Prop:GeneralFnOverArbCQMAC}
The function $f:\ulineCalS \rightarrow \CalS$ of sources $\mathbb{W}_{S_{1}S_{2}}$ can be reconstructed on a CQ-MAC $(\rho_{x_{1}x_{2}} \in \mathcal{D}(\mathcal{H}_{\ulineY}): (x_{1},x_{2}) \in \mathcal{X}_{1}\cross \mathcal{X}_{2})$  if there exists functions $h_j : \CalS_j \rightarrow \CalF_q$ for $j:1,2$, a function $g: \CalF_q \rightarrow \CalS,$ and a PMF $p_{V_1V_2X_1X_2} = p_{V_1X_1}p_{V_2X_2}$ on $\mathcal{V}_1\cross\mathcal{X}_1\cross\mathcal{V}_2\cross\mathcal{X}_2$, where $\mathcal{V}_1 = \mathcal{V}_1 = \mathcal{F}_q$,such that $f(s_1,s_2) = g(h_1(s_1)\oplus h_2(s_2))$ and
$H(h_1(S_1)\oplus_q h_2(S_2)) \leq  \min\{H(V_1),H(V_2)\}\! -H(U) +\chi(\{p_U;\rho_{u}\}),$ where $p_{U}$ and $\rho_{u}$ are as defined in Theorem \ref{thm:SumCQ_MAC}.
\end{prop}
\begin{proof}
The proof follows from the proof of Theorem \ref{thm:SumCQ_MAC}.
\end{proof}
\begin{example}
Let $\mathcal{X}_1=\mathcal{X}_2=\mathcal{S}_1=\mathcal{S}_2=\mathcal{X}=\{0,1\}$, 
$\mathcal{H}=\mathbb{C}^2$,
and 
$\rho_{x_1,x_2}=(1-q)\sigma_{(x_1 \lor x_2)} + q \sigma_{(\bar{x}_1 \land \bar{x}_2)}$, 
where $\sigma_0,\sigma_1 \in \mathcal{D}(\mathcal{H})$ be arbitrary. 
Let $\rho(q) \triangleq (1-q)\sigma_0+q\sigma_1$. 
 Consider correlated symmetric individually uniform sources with $\mathbb{W}_{S_1|S_2}(1|0)=
\mathbb{W}_{S_1|S_2}(0|1)=p$ for $p \in (0,1)$. 
Let $f(S_1,S_2)=S_1 \lor S_2$.
Consider the sufficient conditions given by the unstructured coding scheme: $H(S_1,S_2) <
\chi(\{P_{X_1,X_2},\rho_{x_1,x_2}\})$,
with $X_1$ and $X_2$ being independent,
which can be simplified as
$1+h_b(p) <  S(\rho(0.5))-S(\rho(q))$.
This implies that the $f$ is not reconstructible
using the unstructured codes. 
We embed $f$ in the ternary field. In other words, the encoders and decoder work toward reconstructing 
$S_1 \oplus_3 S_2$. The sufficient condition given by the algebraic coding scheme turns out to be 
\[
H(S_1 \oplus_3 S_2) \!<\! H(X_1\!)-H(X_1 \oplus_3 X_2\!) +
\chi(\!\{p_{X_1 \oplus_3 X_2},\rho_{x_1 \oplus_3 x_2}\}\!),
\]
for some $p_{X_1,X_2}$, 
which can be simplified as 
\begin{align*}
&h_b(2p-p^2)+(2p-p^2)h_b(p/(2-p))
\!< \! \max_\theta [h_b(\theta)\!-\!h_b(2\theta\!-\!\theta^2)-(2\theta\!-\!\theta^2)h_b(\theta/(2-\theta)) +S(\rho((2\theta-\theta^2)*q))-S(\rho(q))]. 
\end{align*}
One can show that there exists choices for $p$, $q$, $\sigma_0$ and $\sigma_1$ such that this condition is satisfied. 
\end{example}
\appendices
\section{Characterization of Certain High Probable Subspaces}
\label{AppSec:CharHighProbSubAKAPinching}
In this appendix, we characterize certain high probability subspaces of tensor product quantum states. The statements we prove here are colloquially referred to as `pinching' \cite{2013Bk_Wil} in the literature. We prove statements in a form that can be used for use in the proof of both Theorems \ref{Thm:NCCAchievesCQPTP}, Lemma \ref{Lem:CharHighProbSubAKAPinching}. We begin with definitions of typical and conditional typical projectors. We adopt strong (frequency) typicality. All statements hold for most of the variants of notion of typicality. For concreteness, the reader may refer to \cite[App.~A]{2020Bk_PraPadShi}.
\begin{lemma}
 \label{Lem:CharHighProbSubAKAPinching}
 Suppose (i) $\mathcal{A},\mathcal{B}$ are finite sets, (ii) $p_{AB}$ is a PMF on $\mathcal{A} \times \mathcal{B}$, (iii) $(\rho_{b} \in \mathcal{D}(\mathcal{H}): b \in \mathcal{B})$ is a collection of density operators, $\rho_{a} \define \sum_{b \in \mathcal{B}}p_{B|A}(b|a)\rho_{b}$ for $a \in \mathcal{A}$ and $\rho = \displaystyle\sum_{a \in \mathcal{A}}p_{A}(a)\rho_{a} = \displaystyle\sum_{b \in \mathcal{B}} p_{B}(b)\rho_{b}$. There exists a strictly positive $\mu>0$, whose value depends only on $p_{AB}$, such that for every $\delta > 0$, there exists a $N(\delta) \in \naturals$ such that for all $n  \geq N(\delta)$, we have
 \begin{eqnarray}
  \label{Eqn:CharHighProbSubAKAPinching-1}
  \tr(\Pi_{\rho}^{\delta}\Pi_{a^{n}}^{\delta}\Pi_{\rho}^{\delta}\rho_{b^{n}}) \geq 1-\exp\{ -n\lambda\delta^{2} \}
  \nonumber
 \end{eqnarray}
whenever $(a^{n},b^{n}) \in T_{\frac{\delta}{4}}^{n}(p_{AB})$ where $\Pi_{a^{n}}^{\delta}$ is the conditional typical projector of $\rho_{a^{n}}= \otimes_{t=1}^{n}\rho_{a_{t}}$ \cite[Defn.~15.2.4]{2013Bk_Wil} and $\Pi_{\rho}^{\delta}$ is the unconditional typical projector \cite[Defn.~15.1.3]{2013Bk_Wil} of $\rho^{\otimes n}$ .
\end{lemma}
\begin{proof}
 We rename $\mathcal{A}=\mathcal{V}$, $\mathcal{B}=\mathcal{X}$, $p_{AB}=p_{VX}$, $a$ as $v$ and $b$ as $x$.
 We have
\begin{align}
\tr(\Pi_{\rho}^{\delta}
\Pi_{v^n}^{\delta} \Pi_{\rho}^{\delta} \rho_{x^n}) &=
\tr(\Pi_{\rho}^{\delta}
\Pi_{v^n}^{\delta} \rho_{x^n} \Pi_{\rho}^{\delta}) \nonumber \\
&\geq \tr(\Pi_{v^n}^{\delta} \rho_{x^n}) -\frac{1}{2} \left\|
\rho_{x^n} -\Pi_{\rho}^{\delta} \rho_{x^n} \Pi_{\rho}^{\delta} \right\|.
\end{align}

In the following we derive a lower bound on $\tr(\Pi_{v^n}^{\delta}\rho_{x^n})$ and
derive an upper bound on $\left\| \rho_{x^n} -\Pi_{\rho}^{\delta} \rho_{x^n} \Pi_{\rho}^{\delta}
\right\|$. Toward the deriving the former, we recall that we have
$(v^n,x^n) \in T_{\delta/2}{^n}(p_{VX})$. Let us define:
\begin{equation}
  p_{Y|XV}(y|x,v) :=\braket{e_{y|v}|\rho_x|e_{y|v}}, \nonumber
  \label{eq:eq(2)}
  \end{equation}
for all $(x,v,y) \in \mathcal{X} \times \mathcal{V} \times \mathcal{Y}$.

Clearly, we have $p_{Y|XV}(y|x,v) \geq 0$, and
$\sum_{y \in \mathcal{Y}} p_{Y|XV}(y|x,v)=\sum_{y \in \mathcal{Y}} \braket{e_{y|v} | \rho_x|
  e_{y|v}}=\tr(\rho_x)=1$. Hence we see that $p_{Y|XV}$ is a stochastic matrix.

Next we note that
\begin{align}
  \sum_{x \in \mathcal{X}} p_{Y|XV}(y|x,v)p_{XV}(x,v) &=
  \sum_{x \in \mathcal{X}} p_{XV}(x,v) \braket{e_{y|v} |\rho_x | e_{y|v}} \nonumber \\
  &= p_{V}(v) \braket{e_{y|v} |\sum_{x \in \mathcal{X}} p_{X|V}(x|v) \rho_x |e_{y|v}} \nonumber \\
  &= p_V(v)\braket{e_{y|v}|\rho_v |e_{y|v}} = p_V(v) q_{Y|V}(y|v),
  \label{equation_1}
  \end{align}
   where we have used the spectral decomposition of $\rho_v$. 

   Observe that  if $(x^n,v^n) \in T_{\delta/4}^n(p_{XV})$, and
   $y^n\in T_{\delta}^n(p_{XV}p_{Y|XV}|x^n,v^n)$, then we have
   $(x^n,v^n,y^n) \in T_{\delta}^n(p_{XV}p_{Y|XV})$. This implies that we have
   $(v^n,y^n) \in T_{\delta}^n(p_{VY})$, where $p_{VY}$ is the marginal
   of $p_{XV}p_{Y|XV}$. Using this and (\ref{equation_1}), we see that
   $(v^n,y^n) \in T_{\delta}^n(p_V q_{Y|V})$. In summary, we see that if
   $(x^n,v^n) \in T_{\delta/4}(p_{XV})$, then we have
   \[
   T_{\delta}^n(p_{XV}p_{Y|XV}|x^n,v^n) \subseteq \left\{y^n: (v^n,y^n) \in
   T_{\delta}^n(p_V q_{Y|V})  \right\}.
   \]
  
   We are now set to provide the promised lower bound. Consider
   \begin{align}
     \tr(\Pi_{v^n}\rho_{x^n}) &= \tr \left(\left[ \sum_{y^n: (v^n,y^n) \in
         T_{\delta}^n(p_Vq_{Y|V})} \bigotimes_{t=1}^n \ket{e_{y_t|v_t}} \bra{e_{y_t|v_t}} \right]
   \left[   \bigotimes_{j=1}^n \rho_{x_j} \right] \right) \\
     &= \tr \left(\left[ \sum_{y^n: (v^n,y^n) \in
       T_{\delta}^n(p_Vq_{Y|V})} \bigotimes_{t=1}^n \ket{e_{y_t|v_t}} \bra{e_{y_t|v_t}}
     \rho_{x_t} \right] \right) \\
   &=\sum_{y^n: (v^n,y^n) \in
     T_{\delta}^n(p_Vq_{Y|V})} \prod_{t=1}^n \braket{e_{y_t|v_t}|\rho_{x_t}|e_{y_t|v_t}} \\
   &\geq \sum_{y^n \in
     T_{\delta}^n(p_{XV}p_{Y|XV}|x^n,v^n))} \prod_{t=1}^n p_{Y|XV}(y_t|x_t,v_t) \\
   &\geq 1- 2|\mathcal{X}||\mathcal{Y}||\mathcal{V}| \exp \left\{-\frac{2n \delta^2
   p_{XVY}(x^*,v^*,y^*)}{4(\log(|\mathcal{X}||\mathcal{Y}||\mathcal{V}|))^2} \right\},
   \end{align}
   where we used the definition (\ref{eq:eq(2)}) in the last equality.
   
   We next provide the upper bound. Note from the Gentle measurements lemma
   \cite[Lemma 9.4.2]{2013Bk_Wil}, we have $\| \rho_{x^n} -\Pi_{\rho}^{\delta}\rho_{x^n}
   \Pi_{\rho}^{\delta}|| \leq 3\sqrt{\epsilon}$ if $\tr(\Pi_{\rho}^{\delta}
   \rho_{x^n})\geq 1-\epsilon$. In the following we provide a lower bound on
  $\tr(\Pi_{\rho}^{\delta}   \rho_{x^n})$. 
   Recall that $\Pi_{\rho}^{\delta} =\sum_{y^n \in T_{\delta}^n(s_Y)} \bigotimes_{t=1}^n
     \ket{g_{y_t}}\bra{g_{y_t}}$, where
     \[
\rho=\sum_{y \in \mathcal{Y}} s_Y(y) \ket{g_y}\bra{g_y},
\]
is the spectral decomposition of $\rho$, and $\rho=\sum_{x \in \mathcal{X}} p_X(x)\rho_x$.
Let $\hat{p}_{Y|X}(y|x):=\braket{g_y|\rho_x|g_y}$, for all
$(x,y)\in \mathcal{X} \times \mathcal{Y}$. Note that $\hat{p}_{Y|X}$ is not
related to $p_{Y|X}$ defined previously. We note that $\hat{p}_{Y|X}(y|x)\geq 0$, and
$\sum_{y \in \mathcal{Y}} \hat{p}_{Y|X}(y|x)=\sum_{y \in \mathcal{Y}} \braket{g_y|\rho_x|g_y}=
\tr(\rho_x)=1$ for all $x \in \mathcal{X}$.
Thus we see that $\hat{p}_{Y|X}$ is a stochastic matrix. It can also be noted that
\[
\sum_{x\in \mathcal{X}} \hat{p}_{Y|X}(y|x)p_X(x) =\braket{g_y|\sum_{x \in \mathcal{X}}
  p_X(x)\rho_x|g_y}=\braket{g_y|\rho|g_y}=s_Y(y),
\]
for all $y \in \mathcal{Y}$. 
This implies that the condition  $y^n \in T_{\delta}^n(s_Y)$ is
equivalent to the condition $y^n \in T_{\delta}^n(\hat{p_{Y}})$,
where $\hat{p_Y}(y)=\sum_{x \in \mathcal{X}} \hat{p}_{Y|X}(y|x)p_X(x)$. Moreover,
if $x^n \in T_{\delta/2}^n(p_X)$, and
$y^n \in T_{\delta}^n(p_X \hat{p}_{Y|X}|x^n)$, then we have
$(x^n,y^n) \in T_{\delta}^n(p_X\hat{p}_{Y|X})$. Consequently, we have
$y^n\in T_{\delta}^n(\hat{p_Y})$, which in turn implies that
$y^n \in T_{\delta}^n(s_Y)$. In essence, we have that if $x^n \in T_{\delta/2}^n(p_X)$
then $T_{\delta}^n(p_X\hat{p}_{Y|X}|x^n) \subseteq T_{\delta}^n(s_Y)$.
Now we are set to provide the lower bound on $\tr(\Pi_{\rho}^{\delta}\rho_{x^n})$ as follows:
\begin{align}
  \tr(\Pi_{\rho}^{\delta} \rho_{x^n}) &= \tr \left(\sum_{y^n \in T_{\delta}(s_Y)} \bigotimes_{t=1}^n
    \ket{g_{y_t}}\bra{g_{y_t}} \rho_{x_t} \right)=\sum_{y^n \in T_{\delta}(s_Y)}
      \prod_{t=1}^n \braket{g_{y_t} |\rho_{x_t} |g_{y_t}} \nonumber \\
      &= \sum_{y^n \in T_{\delta}(s_Y)}
      \prod_{t=1}^n \hat{p}_{Y|X}(y_t|x_t) \geq \sum_{y^n \in T_{\delta}(\hat{p}_{Y|X}p_X|x^n)} 
        \prod_{t=1}^n \hat{p}_{Y|X}(y_t|x_t) \nonumber \\
         &\geq 1-2|\mathcal{X}||\mathcal{Y}| \exp \left\{ -\frac{2n \delta^2 p_X^2(x^*)
          \hat{p}_{Y|X}^2(y|x)}{4(\log(|\mathcal{X}||\mathcal{Y}|))^2} \right\}. 
  \end{align}
We therefore have
\[
\| \rho_{x^n} -\Pi_{\rho}^{\delta} \rho_{x^n} \Pi_{\rho}^{\delta}\| \leq
6 |\mathcal{X}||\mathcal{Y}| \exp \left\{ -\frac{2n \delta^2 p_X^2(x^*)
  \hat{p}_{Y|X}^2(y|x)}{4(\log(|\mathcal{X}||\mathcal{Y}|))^2} \right\},
\]
and
\[
\tr (\Pi_{v^n} \rho_{x^n} ) \geq 1-2|\mathcal{X}||\mathcal{Y}|||\mathcal{V}|
\frac{2n \delta^2 p_X^2(x^*)
          \hat{p}_{Y|X}^2(y|x)}{4(\log(|\mathcal{X}||\mathcal{Y}|))^2},
\]
thereby permitting us to conclude that
\[
\tr (\Pi_{\rho}^{\delta} \Pi_{v^n}^{\delta} \Pi_{\rho}^{\delta} \rho_{x^n}) \geq
\tr (\Pi_{v^n}^{\delta} \rho_{x^n} )-\frac{1}{2} \| \rho_{x^n}-\Pi_{\rho}^{\delta} \rho_{x^n}
\Pi_{\rho}^{\delta} \| \geq 1-\frac{2n \delta^2 p_X^2(x^*)
          \hat{p}_{Y|X}^2(y|x)}{4(\log(|\mathcal{X}||\mathcal{Y}|))^2},
\]
if $(x^n,v^n) \in T_{\delta/2}^n(p_{XV})$. 
\end{proof}

\bibliographystyle{IEEEtran}
{
\bibliography{ComputeOverCQMAC}
\end{document}